\documentclass{article}

\usepackage{microtype}
\usepackage{graphicx}
\usepackage{subfigure}
\usepackage{booktabs} % for professional tables

\usepackage{hyperref}
\newcommand{\alglinelabel}{%
  \addtocounter{ALC@line}{-1}% Reduce line counter by 1
  \refstepcounter{ALC@line}% Increment line counter with reference capability
  \label% Regular \label
}

\usepackage[accepted]{icml2024}

\usepackage{amsmath}
\usepackage{amssymb}
\usepackage{mathtools}
\usepackage{amsthm}

\usepackage[capitalize,noabbrev]{cleveref}

\usepackage{thmtools, thm-restate}

\usepackage{enumitem}
\setlist[itemize]{leftmargin=0.0in}

\usepackage[font=small,aboveskip=1pt,belowskip=1pt]{caption}
\usepackage{microtype}
\newcommand{\myparagraph}[1]{\smallskip\noindent {\bf #1.}}

\theoremstyle{plain}
\newtheorem{theorem}{Theorem}[section]

\theoremstyle{definition}
\newtheorem{definition}[theorem]{Definition}

\theoremstyle{remark}

\usepackage[textsize=tiny]{todonotes}

\icmltitlerunning{Approximate Nearest Neighbor Search with Window Filters}

\DeclareMathOperator*{\argmin}{arg\,min}

\newcommand{\emp}[1]{{\textbf{\textit{#1}}}}
\newcommand{\degree}{\ensuremath{\textit{degree}}}

\newcommand{\datasetname}[1]{\texttt{#1}}
\newcommand{\algname}[1]{\ensuremath{\mathsf{#1}}}

\newcommand*{\Resize}[2]{\resizebox{#1}{!}{$#2$}}%

\usepackage{etoolbox}
\newcommand{\algorithmicdoinparallel}{\textbf{do in parallel}}
\makeatletter
\AtBeginEnvironment{algorithmic}{%
  \newcommand{\PARFOR}[2][default]{\ALC@it\algorithmicfor\ #2\ %
    \algorithmicdoinparallel\ALC@com{#1}\begin{ALC@for}}%
}
\makeatother

\AtBeginEnvironment{quote}{\par\small}

\begin{document}

\twocolumn[
\icmltitle{Approximate Nearest Neighbor Search with Window Filters}



\icmlsetsymbol{equal}{*}

\begin{icmlauthorlist}
\icmlauthor{Joshua Engels}{mit}
\icmlauthor{Benjamin Landrum}{umd}
\icmlauthor{Shangdi Yu}{mit}
\icmlauthor{Laxman Dhulipala
}{umd}
\icmlauthor{Julian Shun}{mit}
\end{icmlauthorlist}

\icmlaffiliation{mit}{Massachusetts Institute of Technology, CSAIL}
\icmlaffiliation{umd}{University of Maryland}

\icmlcorrespondingauthor{Joshua Engels}{jengels@mit.edu}

\icmlkeywords{Approximate Nearest Neighbor Search, Vector Database, Data Structure, Tree, Filtered Search, Embeddings}

\vskip 0.3in
]



\printAffiliationsAndNotice{}  % leave blank if no need to mention equal contribution

\begin{abstract}
We define and investigate the problem of \textit{c-approximate window search}: approximate nearest neighbor search where each point in the dataset has a numeric label, and the goal is to find nearest neighbors to queries within arbitrary label ranges. 
Many semantic search problems, such as image and document search with timestamp filters, or product search with cost filters, are natural examples of this problem. 
We propose and theoretically analyze a modular tree-based framework for transforming an index that solves the traditional c-approximate nearest neighbor problem into a data structure that solves window search. 
On standard nearest neighbor benchmark datasets with random label values, adversarially constructed embeddings, and image search embeddings with real timestamps, we obtain up to a $75\times$ speedup over existing solutions at the same level of recall.
\end{abstract}

\section{Introduction}
The nearest neighbor search problem has been widely studied for more than 30 years~\cite{arya1993approximate}. Given a dataset $D$, the problem requires the construction of an index that can efficiently answer queries of the form ``what is the closest vector to $x$ in $D$?" Solving this problem exactly degrades to a brute force linear search  in high dimensions~\cite{rubinstein2018hardness}, so instead both theoreticians and practitioners focus on the relaxed $c$-approximate nearest neighbor search problem (ANNS), which asks ``what is a vector that is no more than $c$ times farther away from $x$ than the closest vector to $x$ in $D$?"

Recent advances in vector embeddings for text, images, and other unstructured data have led to a  surge in research and industry interest in ANNS. 
This trend is reflected by a proliferation of commercial vector databases: Pinecone~\cite{pinecone2024overview}, Vearch~\cite{vearch2022search}, Weaviate~\cite{weaviate2022filters}, Milvus~\cite{milvus2022hybridsearch}, and Vespa~\cite{vespa2022semistructured}, among others. 
These databases offer hyper-parameter tuning, metadata storage, and solutions to a variety of related search problems. They especially tout their support of ANNS with metadata filtering as a key differentiating application. 
% , where the user wishes to  \laxman{Perhaps you could mention that they drum up filtering as one of the key problems to support for real-world use cases.}

In this work, we are primarily interested in a difficult generalization of the $c$-approximate nearest neighbor problem that we term \emp{Window Search} (see \cref{def:problem}). 
Window search is similar to ANNS, except queries are accompanied by a ``window filter'', and the goal is to find the nearest neighbor to the query that has a label value within the filter. 
This problem has a large number of immediate applications. For instance, in document and image search, each item may be accompanied by a timestamp, and the user may wish to filter to an arbitrary time range (e.g., they may wish to search for hiking pictures, but only from last summer, or forum posts about a bug, but only in the days after a new version was released). Another application is product catalogs, where a user may wish to filter search results by cost. Finally, we note that a large class of emerging applications is large language model retrieval augmented generation, where by storing and retrieving information LLMs can reduce hallucinations and improve their reasoning~\cite{peng2023check}; 
window search may be critical when the LLM needs to recall something stored on a certain date or range of dates. 

Although this problem has many motivating examples, there is a dearth of papers examining it in the literature. 
Some vector databases analyze window search-like problem instances as an additional feature of their system, but this analysis is typically secondary to their main approach and too slow for large-scale real-world systems; as far as we are aware, we are the first to propose, analyze, and experiment with a non-trivial solution to the window search problem. 

% \cref{fig:algorithm} shows a summary of our solution. Each green rectangle represents a subset of $D$ that we have built a normal ANNS index for. By recursively splitting the dataset in half in order of label value and building an ANNS index for each partition, as shown in the top left of the figure, we can query a logarithmic number of indices that combined consider over every point in a query's window filter (see the "Tree Based Query" section of the figure for an example). We also show an alternative type of querying that we call \algname{Optimized Postfiltering}, which entails finding the smallest indexed set that contains all points that meet the filter constraint. This can be efficient, as shown in the "small blowup" case, or inefficient, as shown in the "large blowup" case. Finally, the top right shows a partitioning strategy with overlaps that prevents the large blowup case from happening for all queries. \josh{I'm not 100\% convinced this is the correct place for this, what do other people think?}\shangdi{I think we can delay this to Section 4?} \josh{I'll delete it}

Our contributions include the following:
\begin{enumerate}[topsep=1pt,itemsep=0pt,parsep=0pt,leftmargin=10pt]

    \item We formalize the $c$-approximate window search problem and propose and test the first non-trivial solution that uses the unique nature of numeric label based filters.

    \item We design multiple new algorithms for window search, including a modular tree-based framework and a label-space partitioning approach.

    \item We prove that our tree-based framework solves window search and give runtime bounds for the general case and for a specific instantiation with the Vamana ANNS algorithm~\cite{diskann}. 
    We also analyze optimal partitioning strategies for the label-space partitioning approach.
    
    \item We benchmark our methods against strong baselines and vector databases, achieving up to a $75\times$ speedup on real world embeddings and adversarially constructed datasets. 
\end{enumerate}

% \subsection{Contributions}

\section{Related Work}\label{sec:related}
 
% \julian{define "prefiltering" and "postfiltering" first}
\myparagraph{Filtered ANNS} 
\label{sec:prefiltering_and_postfiltering}
The two overarching naive approaches to filtered ANNS are \emph{prefiltering}, where the dataset is restricted to elements matching the filter before a spatial search over the remaining elements, and \emph{postfiltering}, where results from an unfiltered search are restricted to those matching the filter. Thus, one avenue for solving filtered ANNS has focused on augmenting existing ANNS algorithms with prefiltering or postfiltering strategies, resulting in solutions like VBASE~\cite{vbase} and Milvus~\cite{milvus2022hybridsearch}. VBASE, for example, performs beam search on a general purpose search graph and uses the order in which points are encountered as an approximation of relative distance to the query point, before finally postfiltering to find near points matching the predicate. Non-graph based indices can also be adapted with prefiltering and postfiltering to perform filtered search; for example, the popular Faiss library finds nearby clusters of points within an inverted file index, and then prefilters the vectors in those clusters based on an arbitrary predicate function~\cite{douze2024faiss}. 
While these methods support arbitrary filters, they struggle when filters greatly restrict the points relevant to a query~\cite{filtereddiskANN}. 
% Prior work on filtered ANNS has used either prefiltering or postfiltering with a general purpose index (i.e., Faiss~\cite{douze2024faiss}, Milvus~\cite{milvus2022hybridsearch}, and VBASE~\cite{vbase}) or indices constructed expressly to serve queries filtered with boolean predicates (i.e., FilteredVamana~\cite{filtereddiskANN} and CAPS~\cite{gupta2023caps}). 
% The current state of the art in ANNS is dominated by graph-based methods~\cite{dobson2023scaling}, which are not natively compatible with prefiltering, and struggle with postfiltering when filters greatly restrict the points relevant to a query~\cite{filtereddiskANN} 
% For example, the VBASE vector database performs beam search on a general purpose search graph and uses the order in which points are encountered as the beam is converging as an approximation of relative distance to the query point, before finally  postfiltering to find near points matching the predicate~\cite{vbase}. Non-graph based indices can also be adapted to filtered search; for example, the popular Faiss library finds nearby posting lists within an inverted file index, and then prefilters the vectors in those posting lists based on an arbitrary predicate function~\cite{douze2024faiss}. 

% probably good to mention NHQ but composite indices do not fit elegantly into our characterization here
Another popular direction focuses on \textit{dedicated indices} for filtered ANNS, which consistently outperform their general-purpose counterparts~\cite{neurips23bigann}. For example, FilteredDiskANN, an adaptation of DiskANN~\cite{diskann} that supports filters that are conjunctive ORs of one or more boolean labels, builds a graph that can be traversed by a modified beam search excluding points not matching the filter~\cite{filtereddiskANN}. Another approach is CAPS, which is an inverted index where each partition stores a Huffman tree dividing points by their labels~\cite{gupta2023caps}. While these methods are the state of the art, they are restricted to filtering on boolean labels; they do not support window search. Finally, recent work~\cite{mohoney2023high} presents a tree based nearest neighbor engine for combined vector-metadata searches and show range-based filters as an application, but their partitioning requires historical queries and many of their gains come from batching queries to avoid redundant computation. Their method also only builds ANNS indices in the “leaves” of the tree, whereas we build indices at internal nodes, which is the key idea that ensures our method only needs to query a logarithmic number of ANNS indices. Thus, although their code is not open source, we do not expect their system to perform as well as ours on window search.

% \subsection{Approximate Nearest Neighbor Search (ANNS)}
\myparagraph{Segment Trees}
Segment trees (and the closely related Fenwick
trees~\cite{fenwick1994new}) are tree data structures built over an
array that recursively sub-divide the array to obtain a balanced
binary tree~\cite{bentley1977algorithms}. By storing appropriate augmented values
at the internal nodes of this tree, these data structures can be
used to support a variety of queries over arbitrary intervals
in the array, e.g., computing the maximum value in any given query interval
$[l, r]$.
Segment and Fenwick trees can be generalized to higher
fanout trees, i.e., $B$-ary segment or Fenwick trees that have a
fanout of $B$ and a height of $\lceil \log_{B} n \rceil$~\cite{pibiri2021practical}.
Our work adapts these tree structures to the window search problem by designing a similar data structure that stores ANNS indices at internal tree nodes.
\citet{huang2023faster} use the Fenwick tree for filtered search within a clustering context, but their work only considers a prefix interval ($[0, r]$), and they use $k$d-trees, which are designed for low-dimensional data.

\myparagraph{Filtered Search in Relational Databases}
Traditional relational database systems support arbitrary range-based queries by constructing B-trees or log-structured merge trees~\cite{ilyas2008topk,qader2018comparative}. These databases have complex query planning systems that determine during execution whether and how to query these constructed indices (see, e.g.,~\cite{kurc1999querying}), but typically do not have support for ANNS. The few that do have support for ANNS use an existing open-source ANNS implementation to perform the search (e.g., pgvector~\cite{pgvector}, an ANNS add-on for PostgreSQL, uses HNSW~\cite{malkov2018efficient}.
%\josh{How does this sound? This source seems decent?}

\begin{table}[t]
    \centering
    \caption{Notation used in this paper. 
    }
    \small
    \begin{tabular}{cl}
        \toprule
        Symbol & Definition \\
        \midrule
        $D$ & Vectors to index\\
        $N$ & $|D|$, the cardinality of $D$ \\
        $V$ & Metric space $D$ is in, e.g., $R^n$\\
        $q$ & Query vector $q \in V$ \\
        $(a, b)$ & Window filter; see~\cref{def:filtered_dataset}\\
        $c$ & Approximation factor for window search\\
        $\texttt{dist}_V$ & Distance function between points in $D$\\
        $d$ & Running time to evaluate $\texttt{dist}_V$\\
        $A$ & Arbitrary $c$-ANN algorithm, e.g., Vamana\\
        $A_q$ & Query time of $A$\\
        $\beta$ & Split factor for a $\beta$-WST; see~\cref{alg:build_tree}\\
        $\alpha$ & Pruning parameter for Vamana\\
        $\delta$ & Doubling dimension\\
        $R$ & Set of closed integer ranges in $[1, \ldots, N]$\\
        $\mathsf{blowup}(R)$ & Max ratio of superset $\in R$ over range  length
        \\
        $\mathsf{cost}(R)$ & Sum of lengths of ranges in $R$ \\
        % \midrule
        \bottomrule
    \end{tabular}
    % \vspace{0.1cm}
    \label{tab:notation}
\end{table}

\section{Definitions}

This section lays out definitions for the main problem that we study: window search. Notation for the next three sections can be found in~\cref{tab:notation}.

\begin{definition}
\label{def:labeled-dataset}[Labeled Dataset]
Consider a metric space $V$ with distance function $\texttt{dist}_V$. Given a label function $\ell: V \rightarrow \mathbb{R}$ and a set $D \subset V$, we define a \emp{labeled dataset} to be the pair $(D, \ell)$.
\end{definition}

\begin{definition}
\label{def:filtered_dataset}[Window Filtered Dataset]
Consider a labeled dataset $(D, \ell)$. We define a \emp{window filter} to be an open interval $( a, b)$ with $a, b \in \mathbb{R}$, and we define a \emp{window filtered dataset} to be $D_{(a, b)} = \{x \in D\ |\ \ell(x) \in (a, b)\}$.
\end{definition}

\begin{definition}
\label{def:problem}[Window Search]
Given a labeled dataset $(D, \ell)$, we define a \emp{query} to be a vector $q \in V$ and a window filter $(a, b)$, and we define the \emp{window filtered nearest neighbor} to be
$q^* = \argmin_{x \in D_{(a, b)}} \texttt{dist}_V(x, q)$.
\end{definition}

\begin{definition}\label{def:aws}[Approximate Window Search]
Finally, given a labeled dataset $(D, \ell)$, we define \emp{$c$-approximate window search} to be the task of constructing a data structure that takes in a query $q \in V$ with window filter $(a, b)$ and returns a point $y \in D_{(a, b)}$ such that
$\texttt{dist}_V(q, y) \le c \cdot \texttt{dist}_V(q, q^*)$,
or $\emptyset$ if $D_{(a, b)} = \emptyset$.
% \shangdi{define $\texttt{dist}_V$? why we can't just use $d$?} \josh{Maybe using $\texttt{dist}_V$ makes it more clear its in relation to the metric space $V$? See Definition 3.1 (I just changed it to be $\texttt{dist}_V$, must have lost the underscore V at some point} \shangdi{I see, I got it now. }
\end{definition}

\section{Window Search Algorithms}
In this section, we describe algorithms for solving the window search problem. In \cref{sec:naive_baslines}, we examine two naive baselines for solving window search; in \cref{sec:beta_wst}, we introduce a new data structure, the $\beta$-WST, and design an algorithm to query it; and in \cref{sec:query_methods}, we examine additional algorithms for solving window search. 

We note that with a fixed window filter $(a, b)$, a reasonable approach to solving window search is simply to index $D_{(a, b)}$ using an off-the-shelf ANNS algorithm and then query it with each $q$. Thus, we are interested in the more challenging problem where queries have \textit{arbitrary} window filters.

\subsection{Naive Baselines}
\label{sec:naive_baslines}

As we discuss in \cref{sec:prefiltering_and_postfiltering}, prefiltering and postfiltering are the current state of the art for filtered search. Here, we describe the specific way we implement them in more detail, as they are the main baselines that we compare against in our experiments.

\algname{Prefiltering} is a naive baseline that works by sorting all of the points by label ahead of time. Given a query $x$ with window filter $(a, b)$, we perform binary search on the sorted labels to find the start and end of the range that meet the filter constraints, and then find the distance between $x$ and every point in the range and return the closest point.

\algname{Postfiltering}~\cite{yu2023pecann,chen2020fast} is a second baseline that works by first building an index on all of $D$ using an ANNS algorithm $A$. To perform a window search, we query $A$ for $k \ge 1$ points, and then repeatedly double $k$ until at least one point is returned that has a label within $(a, b)$. Finally, we return the closest of these points. We additionally define a hyper-parameter called $\mathsf{final\_multiply}$; if this is greater than $1$, then we perform a final additional search with a $\mathsf{final\_multiply}$ times larger value of $k$.

\begin{figure}[t]
\centering
\centerline{\includegraphics[width=8.5cm]{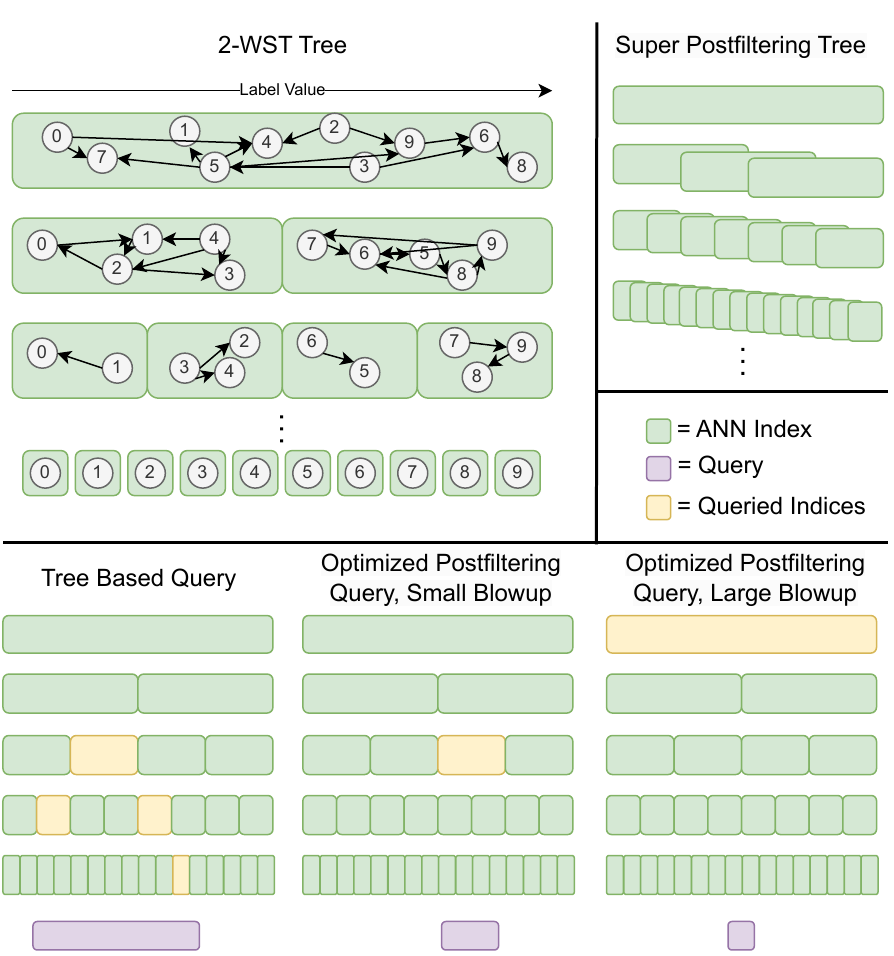}}
\vspace{0.15cm}
\caption{\textbf{Top left}: A $2$-WST. Each node of the tree in green contains a recursive partition of the entire dataset $D$ indexed by an ANNS algorithm (see \cref{alg:build_tree}). The graph in each green node represents a graph-based ANNS index built by an algorithm like Vamana. \textbf{Top right}: the structure of an example label space partitioning method that ensures that no optimized postfiltering query will have a large blowup (see \cref{thm:betterblowup}). \textbf{Bottom}: Different query methods; from left to right: a tree-based query as in \cref{alg:query}, an optimized postfiltering query with a small blowup (see \cref{def:blowup-factor}), and an optimized postfiltering query with a large blowup (see \cref{def:blowup-factor}). 
}
\label{fig:algorithm}
\end{figure}

\subsection{$\beta$-Window Search Tree}
\label{sec:beta_wst}

% Now that we have the problem defined, we can consider algorithms to solve it. 

% More formally, consider a class of ANN search algorithms $A$ that solve the $c$-approximate nearest neighbor ($c$-ANN) problem, and let $A(D_{(a, b)})$ denote the index built from the window filtered dataset $D_{(a,b)}$ such that $A(D_{(a,b)})(q)$ returns a $c$-ANN of $q$. 
% If $A$ solves the traditional $c$-ANN problem, then $A(D_{(a,b)})$ solves the $c$-approximate Window Search problem with the fixed window filter $(a,b)$ with the same time and space bounds. \laxman{The last few sentences could potentially be dropped? It seems like we're just saying that any $c$-ANN solution can be used to solve the problem if the window filter is fixed up front.}

% A more challenging problem is one where we do not know any window filters during indexing, and during querying we receive queries with arbitrary window filters. 
We now propose a data structure that we call a \textbf{$\beta$-Window Search Tree, or $\beta$-WST}. This data structure with $\beta = 2$ is illustrated in the top left of \cref{fig:algorithm} and the accompanying query method is illustrated in the bottom left. The overall idea to construct a $\beta$-WST is to split $D$ into $\beta$ subsets, one corresponding to each child node, construct an instance of $A$ at each node, and recurse on each child. We continue this process until the subset size is less than $\beta$, in which case we just store the points directly. 

More formally, let $x_1, \ldots, x_N$ be the points of $D$ sorted by $\ell$, such that $\ell(x_1) \le \ell(x_2) \ldots \le \ell(x_n)$. With a slight abuse of notation, we define the $\argmin$ of an empty set to be the empty set. A $\beta$-WST works as follows:

\textbf{Index Construction.} A $\beta$-WST $T$ can be constructed as shown in \cref{alg:build_tree}. 
In the base case, if the dataset $D$ is small, we do not construct any tree node (Lines~\ref{alg:build_tree:line:if}--\ref{alg:build_tree:line:return_null}).
% Given a dataset $D$ that is larger than the base case (Line~\ref{alg:build_tree:line:if}), 
Otherwise, we construct an ANNS index of $D$ (Line~\ref{alg:build_tree:line:get_index}). On Lines~\ref{alg:build_tree:line:sizes1}--\ref{alg:build_tree:line:sizes2}, we define the sizes for splitting $D$ into $\beta$ partitions, all of which are equally sized except the last, which may be smaller. On Line~\ref{alg:build_tree:line:initialize_children}, we initialize the children nodes. We then loop through every partition in parallel (Lines~\ref{alg:build_tree:line:parfor}--\ref{alg:build_tree:line:build_tree_recursive}) and recursively call $\texttt{BuildTree}$ on the set of points (sorted by label) corresponding to the start and end of the partition. Finally, on Line~\ref{alg:build_tree:line:return_final} we return the constructed tree, which 
% consists of either just the points $D$ if we execute the base case on line $3$, or otherwise 
is a tuple consisting of an ANNS index built on $D$, the result of $\texttt{BuildTree}$ called on each child along with the size of each child, and the point set $D$. 
 % \shangdi{What is $x$?}
\begin{algorithm}[t]
   \caption{BuildTree$(A, \beta, (D, \ell))$}
   \small
   \label{alg:build_tree}
   \begin{algorithmic}[1]
      \STATE {\bfseries Input:} Dataset $D$ with points $x_1, \ldots, x_{|D|}$ sorted by $\ell$, branching factor $\beta$, ANNS algorithm $A$.
      \label{alg:build_tree:line:input}
      \STATE {\bfseries Output:} $\beta$-Window Search Tree of $D$
      \label{alg:build_tree:line:output}
      \IF{$|D| < \beta$  } \label{alg:build_tree:line:if}
      \STATE \textbf{return} $(\texttt{NULL}, \texttt{NULL}, D)$
      \label{alg:build_tree:line:return_null}
      \ENDIF
      \STATE $\texttt{index} \gets A(D)$

    \COMMENT{The next two lines define the sizes of the $\beta$ subsets; all are equal size except the last}
      
      \label{alg:build_tree:line:get_index}
      \STATE $\texttt{sizes}[1, \ldots, \beta - 1] \gets \lceil |D| / \beta \rceil$
      \label{alg:build_tree:line:sizes1}
      \STATE $\texttt{sizes}[\beta] \gets |D| - (\beta - 1) \cdot \lceil |D| / \beta \rceil$ 
      \label{alg:build_tree:line:sizes2}
      \STATE $\texttt{children}[1, \ldots, \beta] \gets \texttt{NULL}$
      \label{alg:build_tree:line:initialize_children}
      \PARFOR{$i \gets 1$ \textbf{to} $\beta$} \label{alg:build_tree:line:parfor}
          \STATE $\texttt{start} \gets (i - 1) \cdot \lceil |D| / \beta \rceil + 1$
          \label{alg:build_tree:line:start_calculation}
          \STATE $\texttt{end} \gets \texttt{start} + \texttt{sizes}[i]$
          \label{alg:build_tree:line:end_calculation}
          \STATE $\texttt{children}[i] \gets$ BuildTree$(A, \beta, (D_{(x_{\texttt{start}}, x_{\texttt{end} - 1})}, \ell))$ 
          \label{alg:build_tree:line:build_tree_recursive}
      \ENDFOR
      % This is kind of a hack but it is sometimes needed depending on spacing
      % \vspace{-0.37cm} 
      \STATE \textbf{return} $(\texttt{index}, (\texttt{children}, \texttt{sizes}), D)$
      \label{alg:build_tree:line:return_final}
   \end{algorithmic}
\end{algorithm}

\textbf{Querying the Index.} We query a $\beta$-WST using~\cref{alg:query}. The input is a $\beta$-WST $T$ as built by~\cref{alg:build_tree}, a query point $q$ , and a window filter $(a,b)$. At a high level, \cref{alg:query} recurses through the tree constructed by~\cref{alg:build_tree} and queries instances of $A$ that union together to equal the entire filtered dataset. If the dataset is small and the index is a leaf node $\texttt{NULL}$, we do a brute force search over $D$ (Lines~\ref{alg:query:line:leaf}--\ref{alg:query:line:argmin1}). If the window filter $(a,b)$ covers all points, we query the ANNS $\texttt{index}$ and return the result (Lines~\ref{alg:query:line:all}--\ref{alg:query:line:ann}). 
% the algorithm first checks on line $3$ if the $index$ is $\texttt{NULL}$. If so, this is a leaf and we do a brute force search over $D$ on line $4$. 
% Otherwise, we check on line $5$ whether all points in $D$ meet the filter constraints. If so, we just query the ANN $\texttt{index}$ and return the result on line $6$. 
Otherwise, $D$ has some points that do not have a label in $(a, b)$, so we loop over each child (Line~\ref{alg:query:line:loop}) and recurse into it if some of the points in the child meet the query's window filter constraint. Finally, for each one of these children that we recurse into, we add the returned points to a candidate list $\mathcal{L}_{\text{cand}}$, and return the closest point from $\mathcal{L}_{\text{cand}}$ on Line~\ref{alg:query:return}. 

\subsection{Additional Query Methods}
\label{sec:query_methods}

In addition to \cref{alg:query} above, we examine a number of additional query methods for window search that come with various trade-offs.

\begin{list}{\textbullet}{%
    \setlength{\leftmargin}{0.5em}
    \setlength{\itemindent}{0em}
    \setlength{\labelwidth}{\itemindent}
    \setlength{\labelsep}{0.5em}
    \setlength{\listparindent}{1em}
    \setlength{\itemsep}{0em}
    \setlength{\parsep}{0em}
    \setlength{\topsep}{0em}
    \setlength{\partopsep}{0em}
}

    \item \algname{Optimized Postfiltering} uses the index built by \cref{alg:build_tree} but uses a novel query algorithm, shown in \cref{alg:optimized_postfiltering}. Given a query $x$ with window filter $(a, b)$, \algname{Optimized Postfiltering} finds the smallest subset $S$ of $D$ corresponding to an index that we built $I = A(S)$ where $D_{(a, b)} \subset S$, and then queries that index using the same procedure as described in \algname{Postfiltering}. A ``small blowup" query is one in which the smallest subset $S$ is not that much larger than $D_{(a, b)}$, whereas a ``large blowup" query is one where $S$ is much larger than $D_{(a, b)}$. These small and large blowup queries are shown on the bottom right of \cref{fig:algorithm}. Blowup factor is defined formally in \cref{def:blowup-factor}.

    \item \algname{Three Split} also uses the index built by \cref{alg:build_tree} and a novel querying algorithm, shown in \cref{alg:three_split}. Similar to  \cref{alg:query}, a query initially finds the highest level where any partition at all is entirely contained in the window filter, and then does a query on every one of these partitions. Instead of recursing further down the tree, however, \algname{Three Split} then does an \algname{Optimized Postfiltering} call on each of the remaining label subranges on each side of the middle ``covered" label portion. Because we fill in the middle first, we are guaranteed that there can be no ''large blowup" case like in \algname{Optimized Postfiltering}.
    
    \item \algname{Super Postfiltering} is the same as \algname{Optimized Postfiltering} except that it operates on an arbitrary set of indexed subsets of $D$ and not necessarily the ones constructed by \cref{alg:build_tree}. One example of such a data structure is analyzed in \cref{thm:betterblowup} and visualized in the top right of \cref{fig:algorithm}.

\end{list}

\begin{algorithm}[t]
   \caption{Query$(T, q, (a, b))$}
   \small
   \label{alg:query}
\begin{algorithmic}[1]
   \STATE {\bfseries Input:} $\beta$-WST $T = (\texttt{index}, (\texttt{children}, \texttt{sizes}), D)$, with points $x_1, \ldots, x_{|D|} \in D$  sorted by $\ell$, query $q$, window filter $(a, b)$.
   \STATE {\bfseries Output:} Approximate window-filtered nearest neighbor $y$, or $\emptyset$ if no points in $D$ meet window filter constraint.
   \IF{$\texttt{index} = \texttt{NULL}$}\alglinelabel{alg:query:line:leaf}
   \STATE \textbf{return} $\argmin_{y \in D_{(a, b)}} \texttt{dist}_V(q, y)$  \alglinelabel{alg:query:line:argmin1}
   \ENDIF
   \IF{$(\ell(x_1), \ell(x_{|D|})) \subset (a, b)$}\label{alg:query:line:all}
   \STATE \textbf{return} $\texttt{index}(q)$ \label{alg:query:line:ann}
   \ENDIF
    \STATE $\texttt{start} \gets 1$, $\mathcal{L}_{\text{cand}} \gets \emptyset$
   \FOR{$i \gets 1$ \textbf{to} $\beta$ 
   % \julian{should be parallel?} \josh{Maybe not, since in our implementation we parallelize over queries, not within a query. Shangdi mentioned today that arguing why parallellizing individual queries is interesting is a different setting and we would have to make an argument for why it is interesting, so maybe best not to focus on it.}
   } \label{alg:query:line:loop}
   \STATE $\texttt{end} \gets \texttt{start} + \texttt{sizes}[i]$
    \IF{$\ell(x_{\texttt{start}}), \ell(x_{\texttt{end} - 1})\cap (a, b) \ne \emptyset$}
   \STATE $\mathcal{L}_{\text{cand}} \gets \mathcal{L}_{\text{cand}} \cup \text{Query}(\texttt{children}[i], q, (a,b))$ \label{alg:query:line:recurse}
   \ENDIF
   \STATE $\texttt{start} \gets \texttt{end}$
   \ENDFOR
   \STATE \textbf{return} $\argmin_{y \in \mathcal{L}_{\text{cand}}} \texttt{dist}_V(q, y)$ \label{alg:query:return}
\end{algorithmic}
\end{algorithm}

\begin{algorithm}[t]
    \caption{OptimizedPostfiltering$(T, q, (a, b))$}
\small
\label{alg:optimized_postfiltering}
\begin{algorithmic}[1]
   \STATE {\bfseries Input:} $\beta$-WST $T = (\texttt{index}, (\texttt{children}, \texttt{sizes}), D)$, query $q$, window filter $(a, b)$.
   \STATE {\bfseries Output:} Approximate window-filtered nearest neighbor $y$, or $\emptyset$ if no points in $D$ meet window filter constraint.\;
   \STATE $\texttt{index} \gets$ smallest index in $T$ containing all points with labels in $(a, b)$\;
   \STATE $k \gets 1$\;
   \WHILE{$k < N$}
      \STATE $\mathcal{L}_{\text{unfiltered}} \gets \texttt{index}(q, k)$
      \STATE $\mathcal{L}_{\text{cand}} \gets \{x \in \mathcal{L}_{\text{cand}} \; | \: \ell(x) \in (a, b)\}$\;
      \IF{$\mathcal{L}_{\text{cand}} \ne \emptyset$}
        \STATE \textbf{return} $\argmin_{y \in \mathcal{L}_{\text{cand}} \cap D_{(a, b)}} \texttt{dist}_V(q, y)$\;
      \ENDIF
      \STATE $k \gets 2k$\;
   \ENDWHILE
   \STATE \textbf{return} $\emptyset$\;
\end{algorithmic}
\end{algorithm}

\begin{algorithm}[t]
    \caption{ThreeSplit$(T, q, (a, b))$}
\small
\label{alg:three_split}
\begin{algorithmic}[1]
   \STATE {\bfseries Input:} $\beta$-WST $T = (\texttt{index}, (\texttt{children}, \texttt{sizes}), D)$, query $q$, window filter $(a, b)$.
   \STATE {\bfseries Output:} Approximate window-filtered nearest neighbor $y$, or $\emptyset$ if no points in $D$ meet window filter constraint.\;
   \STATE $\texttt{index}, (a', b') \gets$ index in $T$ containing the most points with labels in $(a, b)$, label range of points in \texttt{index}\;
   \STATE $\mathcal{L}_{\text{cand}} = \text{OptimizedPostfiltering}(T, q, (a, a')) \cup \texttt{index}(q) \cup \text{OptimizedPostfiltering}(T, q, (b', b))$\;
    \STATE \textbf{return} $\argmin_{y \in \mathcal{L}_{\text{cand}} \cap D_{(a, b)}} \texttt{dist}_V(q, y)$\;
\end{algorithmic}
\end{algorithm}

\section{Theoretical Analysis}

\subsection{Analysis of Building and Querying a $\beta$-WST}

In our analysis in this section, we assume without loss of generality that $N$ is a power of $\beta$. Removing this assumption would lead to more floor and ceiling operators in \cref{thm:runtime} and leave our other results unchanged. We defer proofs to \cref{sec:proofs}. We will first analyze the construction time and memory of \cref{alg:build_tree}, and then we will prove the correctness and analyze the query time of \cref{alg:query}.

Consider some function $A_f$ parameterized by the dataset $D$, dataset size $N$, and subset size $m$. For example, $A_f$ may be a construction time function or a memory function. This function evaluated on all nodes of a $\beta$-WST is $$O \left( \sum_{j = 0}^{\log_\beta N} \beta^j A_f(D, N \cdot \beta^{-j})\right).$$ 

Since we use Vamana indices in our experiments, we now apply this result to obtain the running time of \cref{alg:build_tree} (the construction time) and the memory of the resulting index. We use a recent analysis for Vamana graph-based search~\cite{indyk2023worst} that assumes a ``slow-preprocessing'' index construction with runtime $O(N^3)$ and memory $O(N (4\alpha)^\delta \log \Delta)$. Letting $\Delta$ be the \textit{aspect ratio} of $D$, i.e., the ratio between the maximum and minimum distances between any two pairs of points, we have the following result:

\begin{restatable}{lemma}{existingannmemory}
\label{lem:existingannmemory}
\cref{alg:build_tree} instantiated with a ``slow preprocessed" $\alpha$-Vamana graph runs in time
$$O\left(\frac{1}{1 - \beta^{-2}} N^3\right) = O(N^3)$$ 
and returns a $\beta$-WST of memory
$$O\left((4\alpha)^\delta \log (\Delta) N \log_{\beta} N\right).$$
\end{restatable}

We now move on to our main runtime theorem, which both proves that \cref{alg:query} indeed solves $c$-approximate window search and upper bounds the running time for an arbitrary ANNS index $A$:

\begin{restatable}{theorem}{runtime}
\label{thm:runtime}
If $A$ can build an index that answers $c$-ANN queries on an arbitrary size $m$ subset of $D$ with query time $O(A_q(D, m))$, and a distance computation in $V$ takes $d$ work, then \cref{alg:query} solves the $c$-approximate window search problem with running time      
$$O\left(\beta\log_\beta(N)d + \beta\sum_{j=0}^{\log_\beta{N}} A_q(D, N \cdot \beta^{-j})\right).$$
\end{restatable}
% \shangdi{$d$ is dimension in Thm 5.1? I think $d$ is overloaded, it is also the distance metric.} \josh{Yes it is overloaded, Theorem 5.1 has an assumption that a distance computation in V takes O(d), while the distance metric in V is $\texttt{dist}_V$. Should we call one of them something else? Idk what letter to choose that would make it more clear}
% \shangdi{oh I got it know. I think it would be good to have a notation table, including $A$, $A_q$, $d$, $V$, $\texttt{dist}_V$, $\beta$, and $D$. We can potentially change $\texttt{dist}_V$ to be $\texttt{dist}_V$, but not a top priority.}

Many theoretical ANN results in the literature have a runtime of $O(N^\rho)$ for a constant $\rho$, e.g., LSH~\cite{andoni2008near} and $k$-nearest neighbor graphs~\cite{prokhorenkova2020graph}. Other results have a runtime that is parameterized only with constants describing the data distribution, and have no reliance on $N$. By applying \cref{thm:runtime}, we have the following results for these common function classes:
\begin{restatable}{lemma}{commonfuncclasses}
\label{lem:commonfuncclasses}
If $A$ is a $c$-ANN algorithm with $A_q(D, m) = O(Cdm^\rho)$ for $\rho \in (0, 1)$ for some constant $C$ depending on $D$, the running time of \cref{alg:query} is
$$O\left(\frac{C\beta dN^\rho}{1 - \beta^{-\rho}}\right).$$
If $A$ is a $c$-ANN algorithm with $O(A_q(D, m)) = O(A_q(D))$, then the running time of \cref{alg:query} is
$$O\left(\beta \log_\beta(N)[d + A_q(D)]\right).$$
\end{restatable}

We can again apply \cref{lem:commonfuncclasses} using recent $c$-ANN results for Vamana graph-based search~\cite{indyk2023worst}. 
%Again letting $\Delta$ be the \textit{aspect ratio} of $D$, 
We have the following result:
\begin{restatable}{lemma}{existingannruntime}
\label{lem:existingannruntime}
\cref{alg:query} instantiated with a ``slow preprocessed" $\alpha$-Vamana graph solves the $c$-approximate window search problem in any metric space on a dataset with doubling dimension $\delta$ and aspect ratio $\Delta$ in running time
$$\Resize{8.3cm}{O\left(\beta\log_\beta(N)\left[d + \log_\alpha\left(\frac{\Delta}{(\alpha - 1 )(c - \frac{\alpha + 1}{\alpha - 1})}\right)(4\alpha)^{\delta} \log \Delta\right]\right)}.$$
\end{restatable}

\subsection{Analysis of Super Postfiltering}

\algname{Super Postfiltering} operates on an arbitrary collection of window filtered subsets $D_{(a_i, b_i)}$. A natural question is how to quantify the quality of a particular choice of subsets to index, which motivates the following definition:

\begin{definition}
\label{def:blowup-factor}[Blowup Factor, Cost]
    Given a set of ranges $R$, with $R_i = [a_i, b_i]$ for $a_i, b_i \in \{1, \ldots, N\}$ and $a_i < b_i$, we can define the \textit{blowup factor} of a new query range $[a, b]$ with $a, b \in \{1, \ldots N\}$ as follows:
    $$\mathsf{blowup}(R, [a, b]) = \min_{\substack{[a_i, b_i] \in R_i \\ [a_i, b_i] \supset [a, b] }} \left(\frac{b_i - a_i}{b - a}\right).$$
    Intuitively, the blowup for $[a, b]$ is the ratio between the size of $[a, b]$ and the smallest range in $R$ that contains it. 
    % Note that if no range $[a_i, b_i]$ contains $[a, b]$, we say that the blowup is equal to $\infty$. \julian{do we need $\infty$ blowup in this paper? if not, we can drop this sentence} 
    We can further define the worst case blowup for a \textit{set of ranges} (like $R$) by taking the maximum blowup over all possible query ranges $[a, b]$: 
    $$\mathsf{blowup}(R) = \max_{\substack{a, b \in \{1,\ldots,N\}   \\ a < b }} \mathsf{blowup}(R, [a, b]).$$
    % \shangdi{This definition is a little confusing for me (I got it now but took some time). Maybe it's helpful to include an example of computing this number. We can put it in appendix and refer to it if there's not enough space.} \josh{Is it clearer now?} 
    We additionally define the \textit{cost} of a set of ranges $R$ as 
    $\mathsf{cost}(R) = \sum_i (b_i - a_i)$. 
\end{definition}

Intuitively, if we build an ANNS index for the points corresponding to each range, then the worst-case blowup limits how expensive a query can be, while the $\mathsf{cost}$ approximates the memory required (since most practical ANNS indices, e.g., LSH~\cite{andoni2008near} and Vamana~\cite{diskann}, have memory that is approximately linear in the number of points that they index). Note that here, a range of, e.g., $[17, 35]$ corresponds to an ANNS index built on the $17$'th point through the $35$'th point in $D$, assuming the points are sorted by label value.

As a quick warmup, we can achieve a worst-case blowup of $N$ and $\mathsf{cost}$ of $N$ by choosing $R = \{[1, N]\}$, and we can get a worst-case blowup of $1$ and $\mathsf{cost}$ of $O(N^3)$ by choosing $R = \{[i, j] \ |\ i, j \in \{1,\ldots,N\}, i < j\}$. We are interested in choices for $R$ that lead to better tradeoffs. 

We can construct an $R$ consisting of the ranges corresponding to each subset indexed in a $\beta$-WST, so we can analyze it using \cref{def:blowup-factor}. We prove the following lemma.

\begin{restatable}{lemma}{wstblowup}
\label{lem:wstblowup}
The ranges corresponding to a $\beta$-WST have worst case blowup factor $B = N / 2 = O(N)$ and  $cost \le N\lceil \log_\beta(N) \rceil = O(N\log_\beta(N))$. 
\end{restatable}

Finally, we prove that we can do better than a $\beta$-WST.

\begin{restatable}{theorem}{betterblowup}
\label{thm:betterblowup}
For any $N$ and for any $\gamma > 1$, there exists an $R$ with worst case blowup factor $2\gamma$ that has cost at most $N\left(2\log_\gamma(N) + 1\right)$.
\end{restatable}

% Sift (1000000, 128)
% Glove (1183514, 100)
% Redcaps (11588824, 512)
% Deep (9990000, 96)
% Aversarial (1000000, 100)

\section{Experiments}

\myparagraph{Experiment Setup}
%\label{sec:experiment_setup}
We run all query methods on all datasets and filter widths on a 2.20GHz Intel Xeon machine with $40$ cores and two-way hyper-threading, $100$ MiB L3 cache, and $504$ GB of RAM. 
We run index building using all $80$ hyper-threads and restrict queries to $16$ threads, and parallelize across (and not within) queries.  
% \ben{You say we run index building on your machine but do you report numbers about index building? added my stuff 2p down}
% We run the index scalability experiments on the same machine, and vary the number of threads as noted in~\cref{sec:varying_index}.
We run index construction experiments with varying values of $\beta$ on a separate 2.10GHz %4 socket 
Intel Xeon machine with $96$ cores and two way hyper-threading, $132$ MiB L3 cache, and $1.47$ TB of RAM. We use all hyper-threads on the machine for these experiments.

We use a Vamana index~\cite{diskann} with $\alpha = 1$, $\degree{} = 64$, and the construction beam search width set to $500$ for all ANNS indices, except for the \algname{Milvus} and \algname{VBASE} baselines. Vamana allows searching for $k \ge 1$ nearest neighbors; for simplicity of presentation, we assume that the query beam search width is always set to $k$. We defer a description of Vamana and its associated hyper-parameters to \cref{app:diskann}.

We note that our theory focuses on the $c$-ANN problem, which only concerns whether a single $c$-approximate nearest neighbor is returned. However, as is standard in much of the ANN literature, in our experiments we report the \textit{recall} of the top $10$ filtered neighbors to the query. While our runtime proofs in
\cref{thm:runtime}, \cref{lem:commonfuncclasses}, and \cref{lem:existingannruntime} 
assume that the underlying ANNS algorithm returns a single ANN, in practice, we find that Vamana has high recall for both single ANN and top-$10$ ANN.
Therefore, we believe that our theoretical analyses still provide insights into our empirical results for top-$10$ ANN.
% but the ideas in our arguments still work.  We assume that if the underlying ANN algorithm has high recall on on single ANN search, then it also has high recall on top-$10$ ANN searches, which we found to be true in practice. 

Finally, we ensure that our smallest filter fractions are still wide enough such that there are at least $10$ points that meet the filter constraint, i.e., we ensure that $|D_{(a, b)}| \ge 10$.

\myparagraph{Implementation Details}
% \shangdi{I think we can merge this section with the experiment section.}
We provide an open source C++ library and associated Python bindings.\footnote{\url{https://github.com/JoshEngels/RangeFilteredANN}} 
% \footnote{https://github.com/JoshEngels/RangeFilteredANN}
Our code is built on the ParlayLib library~\cite{blelloch2020parlaylib} for parallel programming and the recent ParlayANN suite of parallel ANNS algorithms~\cite{dobson2023scaling}. 
We implement a number of memory and performance optimizations, including using a larger base case of $1000$ in \cref{alg:build_tree} and storing the entire dataset just once across all sub-indices. 

% \shangdi{I think we can probably remove this paragraph, since in general there are many performance benefits that C++ has over Python. We can probably put it in the git repo README. }As opposed to a Python-only framework, using C++ gives us more flexibility, allowing us to store the entire dataset just once across all sub-indices \julian{C++ is also a lot faster in general}. This trick reduces the memory from $O(N\cdot L\cdot d + \degree{} \cdot N \cdot L)$, where $L$ is the number of layers, to $O(N \cdot d + \degree{} \cdot N \cdot L)$ in theory, and in practice reduces the memory of the index by a factor of $2$ to $5$. We also are able to parallelize across building different indices and across batches of queries.
% \julian{define what the "degree" is. also, i'm not sure if we defined $d$}

\textbf{Filter Fraction.} Answering a window filter query that matches almost the entire dataset is a substantially different problem than one that matches almost none of it. Thus, we define the \emp{filter fraction} as a way of quantifying where in this filter regime we are: let a query with filter fraction $2^{i}$ for $i \le 0$ be a query whose window filter matches a $2^i$ fraction of the points in $D$. For example, a query with filter fraction $2^{-3}$ has a window filter $(a, b)$ that matches $\frac{1}{8}$ of the dataset, or in other words $|D_{(a, b)}| = \frac{1}{8} |D|$. Queries with a small filter fraction (e.g., $2^{-15}$) restrict to a small portion of the dataset, queries with a large filter fraction (e.g., $2^{-2}$) restrict to a large portion of the dataset, and queries with a medium filter fraction (e.g., $2^{-8}$) restrict somewhere in between. 
%\julian{is the filter fraction random too?} \josh{We create multiple query sets, and for each query set there is a fixed filter fraction.} 
A query with a random filter of fraction $2^{i}$ is a query where we randomly select the filter $(a, b)$ so that all possibilities for the $2^{i} * |D|$ filtered points are equally likely.

\begin{table*}[t]
\footnotesize
    \centering
    \caption{Summary of datasets used in our experiments.}
    \begin{tabular}{|c|c|c|c|c|c|}
        \hline
       Dataset  &  Description & Labels & Num.\ dimensions & Dataset size & Num.\ queries \\
       \hline
        \datasetname{SIFT} & Image feature vectors & Uniform random & $128$ & $1M$ & $10K$\\
        \datasetname{GloVe} & Word embeddings & Uniform random & $100$ & $1.18M$ & $10K$\\
        \datasetname{Deep} & GoogLeNet embeddings & Uniform random & $96$ & $9.9M$ & $10K$\\
        \datasetname{Redcaps} & CLIP image embeddings & Timestamps & $512$ & $11.6M$ & $800$\\
        \datasetname{Adverse} & Mixture of Gaussians & Noisy cluster ID & $100$ & $1M$ & $9.9K$\\
         \hline
    \end{tabular}
    \vspace{0.1cm}
    \label{tab:dataset_sum}
\end{table*}

\myparagraph{Datasets}
%\label{sec:datasets}
The datasets that we use are listed in \cref{tab:dataset_sum} and further explained below. All datasets are available in the same repository as the code; licensing information is included in \cref{app:licensing}.

\begin{list}{\textbullet}{%
    \setlength{\leftmargin}{0.5em}
    \setlength{\itemindent}{0em}
    \setlength{\labelwidth}{\itemindent}
    \setlength{\labelsep}{0.5em}
    \setlength{\listparindent}{1em}
    \setlength{\itemsep}{0em}
    \setlength{\parsep}{0em}
    \setlength{\topsep}{0em}
    \setlength{\partopsep}{0em}
}
    \item \datasetname{SIFT}, \datasetname{GloVe}, and \datasetname{Deep} are ANN datasets from the widely used and standardized ANN benchmarks repository~\cite{aumuller2020ann}. To adapt them to window search, we generate a label for each point uniformly at random between $0$ and $1$. We create $16$ different query sets $Q_1, \ldots, Q_{16}$, each one using the same $10000$ query vectors from ANN benchmarks with random filters of fraction $2^{-i}$. 

    \item \datasetname{Redcaps} is a dataset we created that builds on the RedCaps~\cite{desai2021redcaps} image and caption dataset, which consists of  $11.6M$ Reddit, Imgur, and Flickr images and associated captions. To adapt RedCaps to window search, we use CLIP~\cite{radford2021learning} to generate an embedding for each image and use the timestamp of the image as its label. We create a set of $800$ query vectors by asking ChatGPT-4~\cite{achiam2023gpt} to come up with queries for an image search system, which we then embed using CLIP. See~\cref{app:redcaps_gpt} for full prompt details. We again create $16$ different query sets $Q_1, \ldots, Q_{16}$ using the same $800$ query vectors with random filters of fraction $2^{-i}$. 

    \item \datasetname{Adverse} is a synthetic dataset tailored to disadvantage methods that rely on the label and point distributions being independent. The overall idea is to craft a dataset and queries where the points that meet the filter constraint are much farther away from the query than the rest of the dataset. To do this, we let $D$ be a mixture of $100$ Gaussians and draw $10000$ points from each Gaussian, where the means $\mu_i$ are drawn from $N(0, I)$ and each Gaussian is distributed as $N(\mu_i, 0.01 \cdot I)$ ($I$ is the $100$-dimensional identity matrix). A point generated from the $i$'th Gaussian has a label equal to $i + \mathsf{Uniform}(-0.5, 0.5)$, and we generate a query for every pair $i, j \in \{1,\ldots,100\}$ with $i \ne j$ that consists of a random point drawn from Gaussian $i$ with window filter $(j - 0.5, j + 0.5)$. In other words, each query filters to only points from a different cluster than it itself is drawn from.

\end{list}

% \shangdi{I think we can move the description of optimized postfiltering, super post filtering, and three split to section 4. Since they are part of our contribution. In section 4, we can have a query subsection, and describe different query methods. } \julian{i agree}

\myparagraph{Query Methods and Hyper-parameters}
We run all of the query methods described in \cref{sec:naive_baslines}, \cref{sec:beta_wst}, and \cref{sec:query_methods}. For all methods that use an arbitrary index $A$, we use a Vamana index as described earlier. We run \cref{alg:query} with $\beta = 2$
unless specified otherwise;
%except in~\cref{sec:varying_beta}, where we examine the effect of varying $\beta$; 
we call this method \algname{Vamana WST} in our experiments. We use \algname{Prefiltering} unmodified as described in \cref{sec:naive_baslines}. We expect \algname{Prefiltering} to always achieve near 100\% recall; it may not be 100\% exactly due to numerical precision issues. For \algname{Postfiltering}, \algname{Optimized Postfiltering}, \algname{Three Split}, and \algname{Super Postfiltering}, we search over initial values of $k$ in $[10, 20, 40, 80, 160, 320, 640, 1280]$ and $\mathsf{final\_multiply}$ value in $[1, 2, 3, 4, 8, 16, 32]$. We use the setting $\gamma = 2$  from \cref{thm:betterblowup} for \algname{Super Postfiltering}, which guarantees a worst case blowup factor of $4$ using in practice about $1.5$ times as much memory as \algname{Vamana WST}.

We also compare against \algname{Milvus}~\cite{milvus2022hybridsearch} and \algname{VBASE}~\cite{vbase}, two existing systems that support window search. 

\algname{Milvus} treats window search as an instantiation of categorial filter search. Before querying the underlying ANNS index, \algname{Milvus} creates a bitset that marks all of the points that meet the window filter. Then, while traversing the underlying ANNS index, Milvus ignores points that are not set in the bitset. We try all supported underlying \algname{Milvus} indices: HNSW, IVF\_PQ, IVF\_SQ8, SCANN, and IVF\_FLAT. We search over the same beam sizes as for \algname{Postfiltering}. \algname{Milvus} does not natively support a batch query with different filters for each point in the batch, so we wrote a Python multiprocessing program that spins up many processes that query the constructed index in parallel.

In addition to the heuristic 
% \julian{updated to singular, as I only see one idea for VBASE described in \cref{sec:related}} 
described in \cref{sec:related}, \algname{VBASE} uses a query planning step to attempt to predict how many initial results to retrieve, then filters these retrieved results, before finally applying a final top-$k$ ranking to the retrieved results that meet the filter constraint. We were not able to get multiple queries to run in parallel with VBASE (and in the original paper they also only operate in the regime of a single query at a time).

\begin{figure*}[t]
\begin{center}
\centerline{\includegraphics[width=16cm]{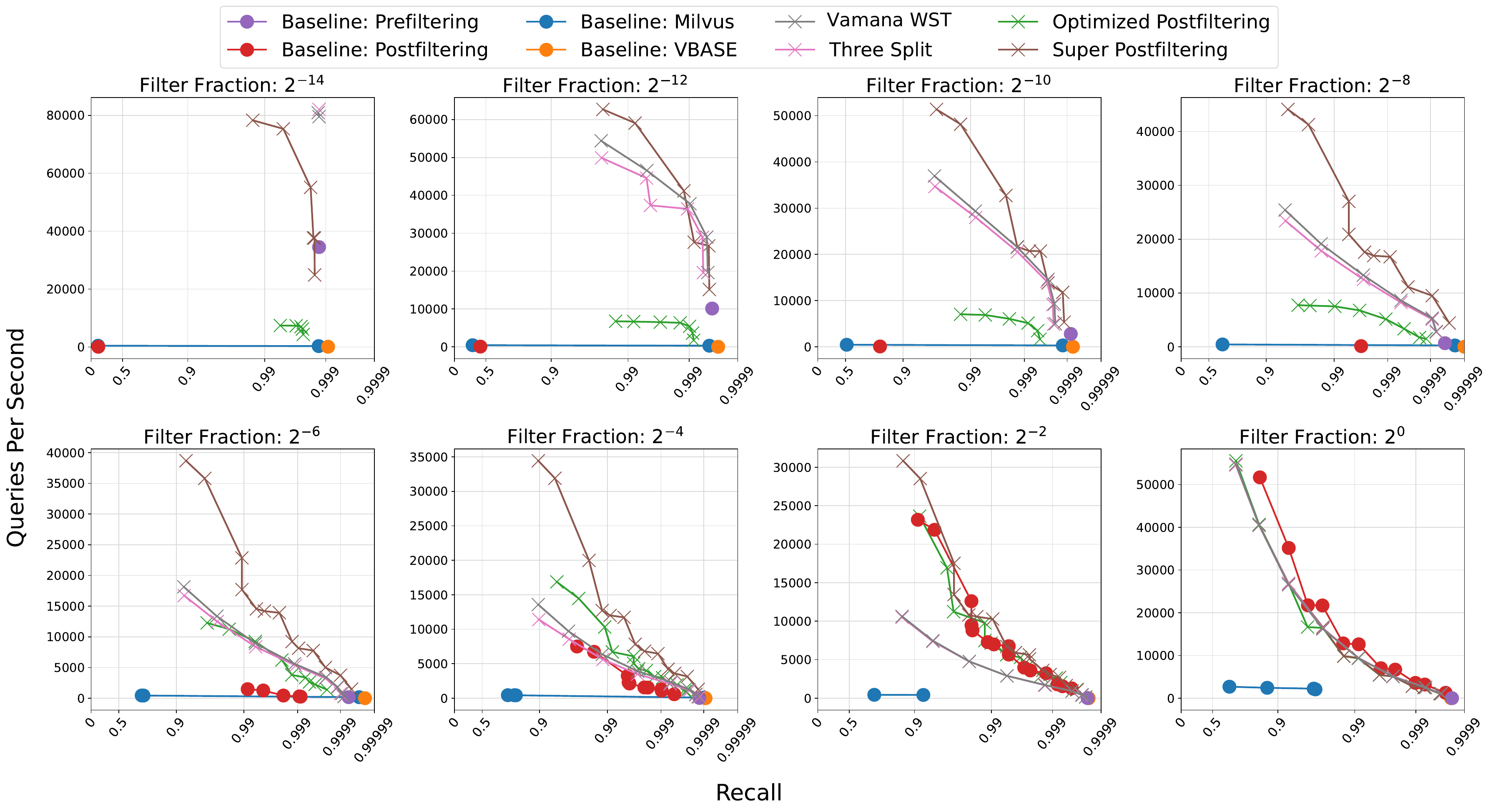}}
\caption{Comparison of Pareto frontiers of all methods on window search with different filter fractions on \datasetname{Deep} using 16 threads. Up and to the right is better. On medium filter fraction settings, our methods achieve orders of magnitude more queries per second than the baselines at the same recall levels. Points along the Pareto frontier are denoted by circles for baseline methods and X's for our methods.
}
\label{fig:deep_results}
\end{center}
%\vspace{-30pt}
\end{figure*}

\myparagraph{Tradeoff of Queries per Second vs.\ Recall}
%\label{sec:qps_v_recall}
We plot the Pareto frontier of recall vs.\ queries per second of all methods on a selection of filter fractions on \datasetname{Deep} in~\cref{fig:deep_results} (the indices we compute the frontier on correspond to the hyper-parameters specified in the last section). 
% For a given method and filter fraction, the Pareto frontier is computed on the set of (recall, queries per second) pairs that come from running an experiment for each hyper-parameter configuration. \julian{it would be helpful to state how the different points for a given method and filter fraction are generated (i'm not sure what parameters there are to tune for the baselines)}
We include plots of other datasets in \cref{sec:experiments-appendix}; the main observations are substantially the same across datasets. Overall, our methods achieve multiple orders of magnitude query speedup at fixed recall levels. \algname{Super Postfiltering} does particularly well at lower recall levels of around $0.9$ to $0.99$, while at high recall levels \algname{Vamana WST} and \algname{Three Split} are the best methods, attaining about the same performance. Additionally, for small filter fractions the \algname{Prefiltering} baseline is competitive with our methods, whereas for large filter fractions the \algname{Postfiltering} baseline is competitive. Our methods achieve the largest speedups over baselines in the medium filter fractions.
% and it is in these domains that we see the largest speedups. 
% \laxman{The reviewers may find the terms ``low'' ``middle'' and ``high'' filter fractions confusing. I'm not sure we clearly say what low/middle/high mean either. How about calling them ``small'', vs. ``large'', where small means the filter is a small fraction of the dataset?}
We note that among our methods, \algname{Optimized Postfiltering} performs the worst, which we explain with the high worst case blowup of a $2$-WST (see~\cref{lem:wstblowup}). Finally, we note that the two vector databases that we test against, \algname{VBASE} and \algname{Milvus}, are completely dominated by our naive baselines \algname{Prefiltering} and \algname{Postfiltering}.

\begin{figure}[t]
\centerline{\includegraphics[width=0.9\columnwidth]{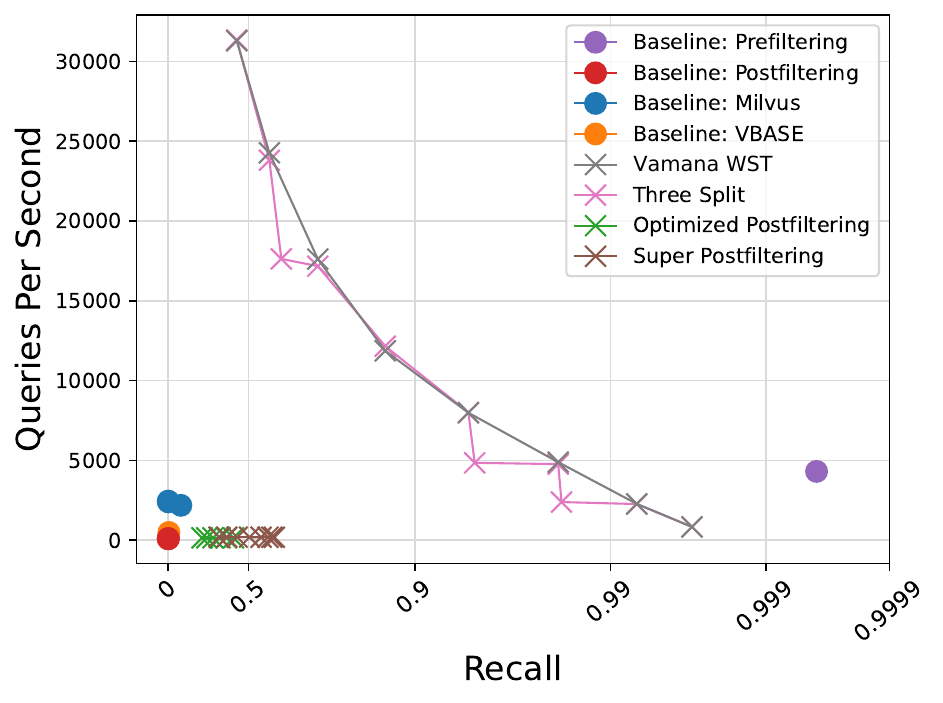}}
\caption{Comparison of window search methods on \datasetname{Adverse}. Up and to the right is better. \algname{Vamana WST} and \algname{Three Split} achieve a good recall vs.\ latency tradeoff, but besides the \algname{Prefiltering} baseline all of the other methods are unable to achieve a reasonable recall. All methods are run with $16$ threads. 
}
\label{fig:adversarial_results}
\end{figure}

We also plot results on \datasetname{Adverse} in~\cref{fig:adversarial_results}.
% , with windows as described in~\cref{sec:datasets}.
\algname{Prefiltering} does well, as expected. Furthermore, \algname{Vamana WST} offers a good tradeoff between recall and query latency, which makes sense because the query time guarantee from~\cref{lem:existingannruntime} (assuming Vamana also does well for top-$10$ ANN) holds for any query distribution, even an adversarial one.
% \julian{the theory only holds for 1-NN though}. \laxman{I think the decomposition result guarantees that we consider all points; the fact that the recall is good is due to Vamana being a good algorithm for $k$-NN empirically?}
However, all of the methods that rely solely on postfiltering, as well as the other baselines, fail to achieve meaningful recall. For the postfiltering methods, this result is unsurprising: the index that gets selected for postfiltering has many points that do not meet the query's filter constraints, and by construction these points are frequently closer to the query than the entire target cluster. The beam search then has to expand over many near neighbors to reach points meeting the filter constraint, at which point the accuracy of the beam search is likely low. Surprisingly, although \algname{Three Split} does use postfiltering as a subroutine, it is able to achieve similar performance to \algname{Vamana WST}; this may be because after querying for the indices that make up the center of the label range (which is \textit{not} postfiltering), the label ranges that are left on each side are typically much smaller. 
% and so are more easily queried or searched by brute force \julian{not sure what "more easily queried" means}. 

\begin{table*}[t]
\footnotesize
\caption{Speedup of our best method over the best baseline, restricted to hyper-parameter settings that yield at least $0.95$ recall. All methods are run with $16$ threads.  We show a speedup across all filter fractions smaller than $2^{-1}$ on all datasets. We show up to a $75\times$ speedup on medium filter fractions.}
\centering
\begin{tabular}{|l|rrrrrrrrrrrr|}
\toprule
Dataset & $2^{-11}$ & $2^{-10}$ & $2^{-9}$ & $2^{-8}$ & $2^{-7}$ & $2^{-6}$ & $2^{-5}$ & $2^{-4}$ & $2^{-3}$ & $2^{-2}$ & $2^{-1}$ & $2^{0}$ \\
\hline
\datasetname{Deep} & 10.49 & 18.46 & 35.65 & 61.21 & 77.55 & 24.28 & 9.35 & 2.67 & 1.46 & 1.39 & 0.75 & 0.77 \\
\datasetname{SIFT} & 1.35 & 1.88 & 3.05 & 4.87 & 8.68 & 16.51 & 11.26 & 4.46 & 2.26 & 1.28 & 0.90 & 0.92 \\
\datasetname{GloVe} & 1.90 & 2.27 & 2.70 & 3.77 & 5.02 & 6.19 & 9.60 & 7.62 & 2.34 & 1.52 & 0.92 & 0.92 \\
\datasetname{Redcaps} & 2.31 & 3.33 & 5.47 & 7.78 & 10.07 & 17.22 & 3.94 & 3.64 & 1.75 & 1.73 & 0.90 & 0.90 \\
\hline
\end{tabular}
\vspace{0.1cm}
\label{tab:speedups}
\end{table*}

We also include speedups of our best method (the best of \algname{Vamana WST}, \algname{Optimized Postfiltering}, \algname{Super Postfiltering}, and \algname{Three Split}) over the best baseline method on filter fractions $i = 2^0, \ldots, 2^{-11}$ on all datasets in~\cref{tab:speedups} at a recall level of $0.95$, and we include the same table at additional recall levels in \cref{sec:experiments-appendix}. These reinforce our findings in~\cref{fig:deep_results} across other datasets; at a $0.95$ recall level, we achieve up to a $75$X speedup on \datasetname{Deep}, up to a $16$X speedup on \datasetname{SIFT}, up to a $9$X speedup on \datasetname{GloVe}, and up to a $17$X speedup on \datasetname{Redcaps}.

\begin{table}[t]
\footnotesize
\caption{Build times for different indexing methods across all datasets, rounded to the nearest unit. Index construction was done using $80$ threads.}
\vspace{0.1cm}
\centering
\begin{tabular}{|l|r|r|r|r|}
\hline
\textbf{Dataset} & \textbf{Vamana} & \textbf{$2$-WST} & \textbf{Super} \\
\hline
\datasetname{SIFT}  & 1 m & 8 m & 14 m\\
\datasetname{GloVe}  & 3 m & 17 m & 28 m \\
\datasetname{Deep}  & 17 m & 2 h & 4 h  \\
\datasetname{Redcaps} & 2 h & 7 h & 19 h \\
\datasetname{Adverse} & 3 m & 23 m & 41 m \\
\hline
\end{tabular}
\vspace{0.1cm}
\label{table:build-times-datasets}
\end{table}

\begin{table}[t]
\centering

\caption{Index sizes for different indexing methods on all datasets, rounded to the nearest $10$th of a GB. Note that prefiltering takes just the memory of the original dataset. The "Raw" column is the size of just the dataset.}
\footnotesize
\begin{tabular}{|l|r|r|r|r|r|}
\hline
\textbf{Dataset} & \textbf{Raw} & \textbf{Vamana} & \textbf{$2$-WST} & \textbf{Super} \\
\hline
\datasetname{SIFT}  & 0.5 GB & 1.0 GB &  4.7 GB & 7.6 GB\\
\datasetname{GloVe}  & 0.5 GB & 1.1 GB & 5.6 GB & 9.5 GB \\
\datasetname{Deep} & 3.6 GB & 6.8 GB &  53.2 GB&  94.6 GB\\
\datasetname{Redcaps} & 23 GB & 27.1 GB & 79.2 GB & 127 GB \\
\datasetname{Adverse} & 0.8 GB & 0.9 GB & 4.6 GB & 7.4 GB \\
\hline
\end{tabular}
\vspace{0.1cm}
\label{table:index-sizes-datasets}
\end{table}

\myparagraph{Index Memory and Construction Time}
%\label{sec:varying_index}
\cref{table:build-times-datasets} and \cref{table:index-sizes-datasets} show the construction time and memory sizes for a single Vamana index (for \algname{Postfiltering}), a $2$-WST with Vamana indices at each node (for \algname{Vamana WST}, \algname{Optimized Postfiltering}, and \algname{Three Split}), and the index for \algname{Super Postfiltering}. We also show the memory for just the dataset, which is exactly how much memory \algname{Prefiltering} needs (the build for \algname{Prefiltering} is just a sort and takes less than a few seconds across all datasets). We see that the $2$-WST takes about $3$--$8$ times as much memory as a single Vamana index, depending on the dataset, whereas the index for \algname{Super Postfiltering} takes about $5$--$14$ times as much memory as a single Vamana index. The construction times show a similar increase across methods, with a $5$--$10$X increase in build time from Vamana to $2$-WST and a $10$--$20$X increase from Vamana to \algname{Super Postfiltering}. Finally, we note that while the larger datasets do take signficantly longer to build, we are using a large beam search construction buffer of $500$ in order to vary as few hyper-parameters as possible, and a smaller buffer size may give faster build times with minimal loss in recall.

\begin{figure}
    \centering
    \includegraphics[width=0.7\linewidth]{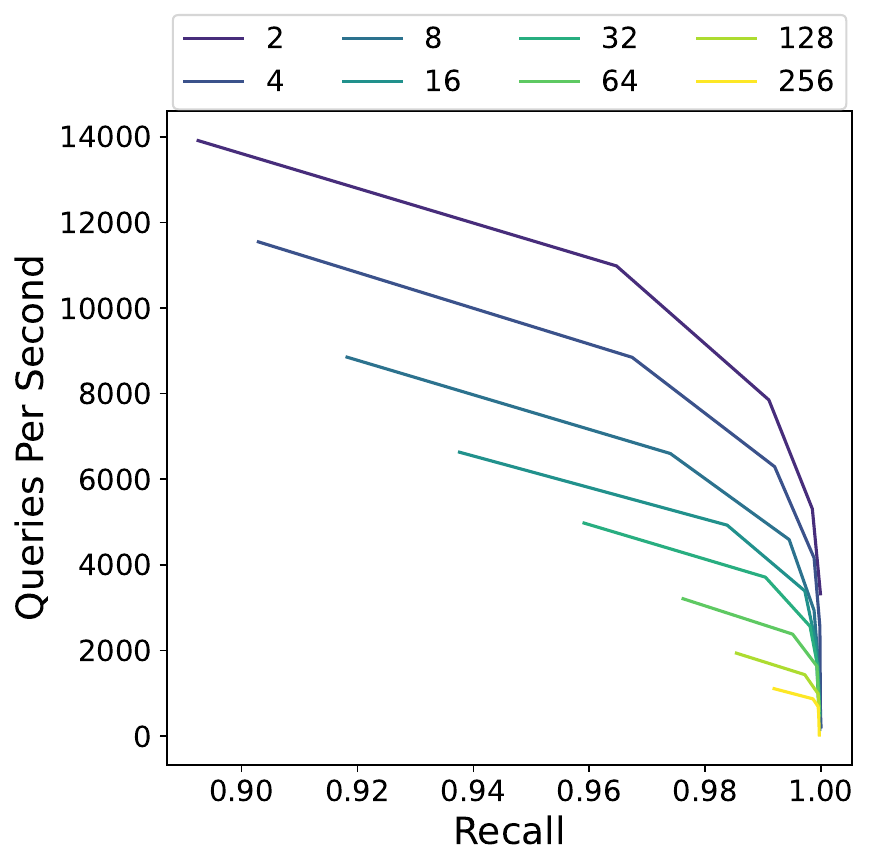}
    \caption{Pareto curves of recall vs.\ throughput on the SIFT dataset for a filter fraction of $2^{-1}$ and varying branching factors $\beta$ for \algname{Vamana WST}. Up and to the right is better. All trials are run with $16$ threads.}
    \label{fig:half_branching_filter_widths}
\end{figure}

\myparagraph{Varying $\beta$}
%\label{sec:varying_beta}
A larger branching factor $\beta$ decreases the build time and memory footprint of a \algname{Vamana WST} index by reducing the number of levels in the tree and thus reducing the number of graphs that are built (see~\cref{fig:branching_buildtime_memory} in \cref{sec:experiments-appendix} for a plot showing the exact reduction in memory and build time as we  scale $\beta$). However, as shown in~\cref{fig:half_branching_filter_widths}, this comes at the cost of a reduction in recall and latency. These results are substantially the same across different filter fractions; see~\cref{fig:full_branching_filter_widths} in \cref{sec:experiments-appendix} for experiments with more filter fractions.

% \includegraphics[width=\textwidth]{images/sift-128-euclidean_results.pdf}

% \includegraphics[width=\textwidth]{images/deep-image-96-angular_results.pdf}

% \includegraphics[width=\textwidth]{images/redcaps-512-angular_results.pdf}
% \section{Conclusion}

\section{Conclusion}
We identify window search as an important and overlooked search problem. The methods we present for solving window search give a significant speedup over the state of the art, have strong theoretical guarantees, and are an important step towards ensuring vector databases efficiently support a full range of necessary embedding search operations.

% \subsection{Limitations}
% \julian{we already discuss this earlier. i think we can remove this subsection}
% The main limitation of the methods we propose are that they require an increased index size and increased construction time over a naive \algname{Prefiltering} or \algname{Postfiltering} index. This is an example of the common memory-performance tradeoff; because the memory and construction time increases are only roughly a constant amount larger than a single index, we still expect our method to be useful.

% \clearpage
% Can be on its own page at the end
\section*{Impact Statement}

We expect the primary broader impact of our work to be general improvement of semantic vector search. Thus, we do not expect our work to have societal consequences beyond generally improving the accuracy and performance of machine learning systems. These systems have many potential societal consequences, none of which we feel must be specifically highlighted here.

\section*{Acknowledgements}
 This research is supported by
DOE Early Career Award \#DE-SC0018947,
NSF CAREER Award \#CCF-1845763, NSF Award \#CCF-2316235, NSF Award \#CNS-2317194,
Google Faculty Research Award, Google Research Scholar Award, and
FinTech@CSAIL Initiative.

\bibliography{paper}

\begin{thebibliography}{37}
\providecommand{\natexlab}[1]{#1}
\providecommand{\url}[1]{\texttt{#1}}
\expandafter\ifx\csname urlstyle\endcsname\relax
  \providecommand{\doi}[1]{doi: #1}\else
  \providecommand{\doi}{doi: \begingroup \urlstyle{rm}\Url}\fi

\bibitem[mil(2022)]{milvus2022hybridsearch}
Milvus-docs: Conduct a hybrid search, 2022.
\newblock URL \url{https://github.com/milvus-io/milvus-docs/blob/v2.1.x/site/en/userGuide/search/hybridsearch.md}.

\bibitem[vea(2022)]{vearch2022search}
Vearch doc operation: Search, 2022.
\newblock URL \url{https://vearch.readthedocs.io/en/latest/use_op/op_doc.html?highlight=filter#search}.

\bibitem[ves(2022)]{vespa2022semistructured}
Vespa use cases: Semi-structured navigation, 2022.
\newblock URL \url{https://docs.vespa.ai/en/attributes.html}.

\bibitem[wea(2022)]{weaviate2022filters}
Weaviate documentation: Filters, 2022.
\newblock URL \url{https://weaviate.io/developers/weaviate/current/graphql-references/filters.html}.

\bibitem[Achiam et~al.(2023)Achiam, Adler, Agarwal, Ahmad, Akkaya, Aleman, Almeida, Altenschmidt, Altman, Anadkat, et~al.]{achiam2023gpt}
Achiam, J., Adler, S., Agarwal, S., Ahmad, L., Akkaya, I., Aleman, F.~L., Almeida, D., Altenschmidt, J., Altman, S., Anadkat, S., et~al.
\newblock {GPT}-4 technical report.
\newblock \emph{arXiv preprint arXiv:2303.08774}, 2023.

\bibitem[Andoni \& Indyk(2008)Andoni and Indyk]{andoni2008near}
Andoni, A. and Indyk, P.
\newblock Near-optimal hashing algorithms for approximate nearest neighbor in high dimensions.
\newblock \emph{Communications of the ACM}, 51\penalty0 (1):\penalty0 117--122, 2008.

\bibitem[Arya \& Mount(1993)Arya and Mount]{arya1993approximate}
Arya, S. and Mount, D.~M.
\newblock Approximate nearest neighbor queries in fixed dimensions.
\newblock In \emph{ACM-SIAM Symposium on Discrete Algorithms}, pp.\  271--280, 1993.

\bibitem[Aum{\"u}ller et~al.(2020)Aum{\"u}ller, Bernhardsson, and Faithfull]{aumuller2020ann}
Aum{\"u}ller, M., Bernhardsson, E., and Faithfull, A.
\newblock {ANN}-benchmarks: A benchmarking tool for approximate nearest neighbor algorithms.
\newblock \emph{Information Systems}, 87:\penalty0 101374, 2020.

\bibitem[Bentley(1977)]{bentley1977algorithms}
Bentley, J.~L.
\newblock Algorithms for {Klee}’s rectangle problems.
\newblock Technical report, Technical Report, 1977.

\bibitem[Blelloch et~al.(2020)Blelloch, Anderson, and Dhulipala]{blelloch2020parlaylib}
Blelloch, G.~E., Anderson, D., and Dhulipala, L.
\newblock {ParlayLib}--a toolkit for parallel algorithms on shared-memory multicore machines.
\newblock In \emph{ACM Symposium on Parallelism in Algorithms and Architectures}, pp.\  507--509, 2020.

\bibitem[Chen et~al.(2020)Chen, Hu, Fan, Shen, Zhang, Liu, Du, Li, Chen, and Li]{chen2020fast}
Chen, Y., Hu, X., Fan, W., Shen, L., Zhang, Z., Liu, X., Du, J., Li, H., Chen, Y., and Li, H.
\newblock Fast density peak clustering for large scale data based on {kNN}.
\newblock \emph{Knowledge-Based Systems}, 187:\penalty0 104824, 2020.

\bibitem[Clarkson(1997)]{clarkson1997nearest}
Clarkson, K.~L.
\newblock Nearest neighbor queries in metric spaces.
\newblock In \emph{ACM Symposium on Theory of Computing}, pp.\  609--617, 1997.

\bibitem[Desai et~al.(2021)Desai, Kaul, Aysola, and Johnson]{desai2021redcaps}
Desai, K., Kaul, G., Aysola, Z., and Johnson, J.
\newblock {RedCaps}: Web-curated image-text data created by the people, for the people.
\newblock In \emph{Advances in Neural Information Processing Systems}, 2021.

\bibitem[Douze et~al.(2024)Douze, Guzhva, Deng, Johnson, Szilvasy, Mazaré, Lomeli, Hosseini, and Jégou]{douze2024faiss}
Douze, M., Guzhva, A., Deng, C., Johnson, J., Szilvasy, G., Mazaré, P.-E., Lomeli, M., Hosseini, L., and Jégou, H.
\newblock The {Faiss} library.
\newblock \emph{arXiv e-prints}, 2024.

\bibitem[Fenwick(1994)]{fenwick1994new}
Fenwick, P.~M.
\newblock A new data structure for cumulative frequency tables.
\newblock \emph{Software: Practice and Experience}, 24\penalty0 (3):\penalty0 327--336, 1994.

\bibitem[Gollapudi et~al.(2023)Gollapudi, Karia, Sivashankar, Krishnaswamy, Begwani, Raz, Lin, Zhang, Mahapatro, Srinivasan, Singh, and Simhadri]{filtereddiskANN}
Gollapudi, S., Karia, N., Sivashankar, V., Krishnaswamy, R., Begwani, N., Raz, S., Lin, Y., Zhang, Y., Mahapatro, N., Srinivasan, P., Singh, A., and Simhadri, H.~V.
\newblock Filtered-{DiskANN}: Graph algorithms for approximate nearest neighbor search with filters.
\newblock In \emph{ACM Web Conference}, pp.\  3406–3416, 2023.

\bibitem[Gupta et~al.(2023)Gupta, Yi, Coleman, Luo, Lakshman, and Shrivastava]{gupta2023caps}
Gupta, G., Yi, J., Coleman, B., Luo, C., Lakshman, V., and Shrivastava, A.
\newblock {CAPS}: A practical partition index for filtered similarity search, 2023.

\bibitem[Huang et~al.(2023)Huang, Yu, and Shun]{huang2023faster}
Huang, Y., Yu, S., and Shun, J.
\newblock Faster parallel exact density peaks clustering.
\newblock In \emph{SIAM Conference on Applied and Computational Discrete Algorithms (ACDA23)}, pp.\  49--62. SIAM, 2023.

\bibitem[Ilyas et~al.(2008)Ilyas, Beskales, and Soliman]{ilyas2008topk}
Ilyas, I.~F., Beskales, G., and Soliman, M.~A.
\newblock A survey of top-k query processing techniques in relational database systems.
\newblock \emph{ACM Comput. Surv.}, 40\penalty0 (4), oct 2008.

\bibitem[Indyk \& Xu(2023)Indyk and Xu]{indyk2023worst}
Indyk, P. and Xu, H.
\newblock Worst-case performance of popular approximate nearest neighbor search implementations: Guarantees and limitations.
\newblock In \emph{Advances in Neural Information Processing Systems}, 2023.

\bibitem[Jayaram~Subramanya et~al.(2019)Jayaram~Subramanya, Devvrit, Simhadri, Krishnawamy, and Kadekodi]{diskann}
Jayaram~Subramanya, S., Devvrit, F., Simhadri, H.~V., Krishnawamy, R., and Kadekodi, R.
\newblock {DiskANN}: Fast accurate billion-point nearest neighbor search on a single node.
\newblock In \emph{Advances in Neural Information Processing Systems}, 2019.

\bibitem[Kane(2024)]{pgvector}
Kane, A.
\newblock pgvector: Open-source vector similarity search for postgres.
\newblock \url{https://github.com/pgvector/pgvector}, 2024.
\newblock Accessed: 2024-05-24.

\bibitem[Krauthgamer \& Lee(2004)Krauthgamer and Lee]{krauthgamer2004navigating}
Krauthgamer, R. and Lee, J.~R.
\newblock Navigating nets: simple algorithms for proximity search.
\newblock In \emph{ACM-SIAM Symposium on Discrete Algorithms}, pp.\  798--807, 2004.

\bibitem[Kurc et~al.(1999)Kurc, Chang, Ferreira, Sussman, and Saltz]{kurc1999querying}
Kurc, T., Chang, C., Ferreira, R., Sussman, A., and Saltz, J.
\newblock Querying very large multi-dimensional datasets in {ADR}.
\newblock In \emph{ACM/IEEE Conference on Supercomputing}, 1999.

\bibitem[Malkov \& Yashunin(2018)Malkov and Yashunin]{malkov2018efficient}
Malkov, Y.~A. and Yashunin, D.~A.
\newblock Efficient and robust approximate nearest neighbor search using hierarchical navigable small world graphs.
\newblock \emph{IEEE Transactions on Pattern Analysis and Machine Intelligence}, 42\penalty0 (4):\penalty0 824--836, 2018.

\bibitem[Manohar et~al.(2024)Manohar, Shen, Blelloch, Dhulipala, Gu, Simhadri, and Sun]{dobson2023scaling}
Manohar, M.~D., Shen, Z., Blelloch, G., Dhulipala, L., Gu, Y., Simhadri, H.~V., and Sun, Y.
\newblock {ParlayANN}: Scalable and deterministic parallel graph-based approximate nearest neighbor search algorithms.
\newblock In \emph{ACM SIGPLAN Annual Symposium on Principles and Practice of Parallel Programming}, pp.\  270–285, 2024.

\bibitem[Mohoney et~al.(2023)Mohoney, Pacaci, Chowdhury, Mousavi, Ilyas, Minhas, Pound, and Rekatsinas]{mohoney2023high}
Mohoney, J., Pacaci, A., Chowdhury, S.~R., Mousavi, A., Ilyas, I.~F., Minhas, U.~F., Pound, J., and Rekatsinas, T.
\newblock High-throughput vector similarity search in knowledge graphs.
\newblock \emph{Proceedings of the ACM on Management of Data}, 1\penalty0 (2):\penalty0 1--25, 2023.

\bibitem[Peng et~al.(2023)Peng, Galley, He, Cheng, Xie, Hu, Huang, Liden, Yu, Chen, et~al.]{peng2023check}
Peng, B., Galley, M., He, P., Cheng, H., Xie, Y., Hu, Y., Huang, Q., Liden, L., Yu, Z., Chen, W., et~al.
\newblock Check your facts and try again: Improving large language models with external knowledge and automated feedback.
\newblock \emph{arXiv preprint arXiv:2302.12813}, 2023.

\bibitem[Pibiri \& Venturini(2021)Pibiri and Venturini]{pibiri2021practical}
Pibiri, G.~E. and Venturini, R.
\newblock Practical trade-offs for the prefix-sum problem.
\newblock \emph{Software: Practice and Experience}, 51\penalty0 (5):\penalty0 921--949, 2021.

\bibitem[{Pinecone Systems, Inc.}(2024)]{pinecone2024overview}
{Pinecone Systems, Inc.}
\newblock Overview, 2024.
\newblock URL \url{https://docs.pinecone.io/docs/overview}.

\bibitem[Prokhorenkova \& Shekhovtsov(2020)Prokhorenkova and Shekhovtsov]{prokhorenkova2020graph}
Prokhorenkova, L. and Shekhovtsov, A.
\newblock Graph-based nearest neighbor search: From practice to theory.
\newblock In \emph{International Conference on Machine Learning}, pp.\  7803--7813, 2020.

\bibitem[Qader et~al.(2018)Qader, Cheng, and Hristidis]{qader2018comparative}
Qader, M., Cheng, S., and Hristidis, V.
\newblock A comparative study of secondary indexing techniques in {LSM}-based {NoSQL} databases.
\newblock In \emph{International Conference on Management of Data}, pp.\  551--566, 2018.

\bibitem[Radford et~al.(2021)Radford, Kim, Hallacy, Ramesh, Goh, Agarwal, Sastry, Askell, Mishkin, Clark, et~al.]{radford2021learning}
Radford, A., Kim, J.~W., Hallacy, C., Ramesh, A., Goh, G., Agarwal, S., Sastry, G., Askell, A., Mishkin, P., Clark, J., et~al.
\newblock Learning transferable visual models from natural language supervision.
\newblock In \emph{International Conference on Machine Learning}, pp.\  8748--8763, 2021.

\bibitem[Rubinstein(2018)]{rubinstein2018hardness}
Rubinstein, A.
\newblock Hardness of approximate nearest neighbor search.
\newblock In \emph{ACM SIGACT Symposium on Theory of Computing}, pp.\  1260--1268, 2018.

\bibitem[Simhadri et~al.(2023)Simhadri, Aumüller, Baranchuk, Douze, Ingber, Liberty, and Williams]{neurips23bigann}
Simhadri, H., Aumüller, M., Baranchuk, D., Douze, M., Ingber, A., Liberty, E., and Williams, G.
\newblock {NeurIPS}'23 competition track: Big-{ANN}, 2023.
\newblock URL \url{https://big-ann-benchmarks.com/neurips23.html}.

\bibitem[Yu et~al.(2023)Yu, Engels, Huang, and Shun]{yu2023pecann}
Yu, S., Engels, J., Huang, Y., and Shun, J.
\newblock Pecann: Parallel efficient clustering with graph-based approximate nearest neighbor search.
\newblock \emph{arXiv preprint arXiv:2312.03940}, 2023.

\bibitem[Zhang et~al.(2023)Zhang, Xu, Chen, Sui, Xie, Cai, Chen, He, Yang, Yang, Yang, and Zhou]{vbase}
Zhang, Q., Xu, S., Chen, Q., Sui, G., Xie, J., Cai, Z., Chen, Y., He, Y., Yang, Y., Yang, F., Yang, M., and Zhou, L.
\newblock {VBASE}: Unifying online vector similarity search and relational queries via relaxed monotonicity.
\newblock In \emph{USENIX Symposium on Operating Systems Design and Implementation (OSDI)}, pp.\  377--395, 2023.

\end{thebibliography}
\bibliographystyle{icml2024}

%%%%%%%%%%%%%%%%%%%%%%%%%%%%%%%%%%%%%%%%%%%%%%%%%%%%%%%%%%%%%%%%%%%%%%%%%%%%%%%
%%%%%%%%%%%%%%%%%%%%%%%%%%%%%%%%%%%%%%%%%%%%%%%%%%%%%%%%%%%%%%%%%%%%%%%%%%%%%%%
% APPENDIX
%%%%%%%%%%%%%%%%%%%%%%%%%%%%%%%%%%%%%%%%%%%%%%%%%%%%%%%%%%%%%%%%%%%%%%%%%%%%%%%
%%%%%%%%%%%%%%%%%%%%%%%%%%%%%%%%%%%%%%%%%%%%%%%%%%%%%%%%%%%%%%%%%%%%%%%%%%%%%%%

%%%%%%%%%%%%%%%%%%%%%%%%%%%%%%%%%%%%%%%%%%%%%%%%%%%%%%%%%%%%%%%%%%%%%%%%%%%%%%%
%%%%%%%%%%%%%%%%%%%%%%%%%%%%%%%%%%%%%%%%%%%%%%%%%%%%%%%%%%%%%%%%%%%%%%%%%%%%%%%
% APPENDIX
%%%%%%%%%%%%%%%%%%%%%%%%%%%%%%%%%%%%%%%%%%%%%%%%%%%%%%%%%%%%%%%%%%%%%%%%%%%%%%%
%%%%%%%%%%%%%%%%%%%%%%%%%%%%%%%%%%%%%%%%%%%%%%%%%%%%%%%%%%%%%%%%%%%%%%%%%%%%%%%
\newpage
\appendix
\onecolumn

\section{Proofs}\label{sec:proofs}

\existingannmemory
\begin{proof}

As stated in the main text, if we have some function parameterized by the dataset and subset size $O(A_f(D, m))$, then this function evaluated on all nodes of \cref{alg:build_tree} is

\begin{align}
\label{eqn:helpful}
O \left( \sum_{j = 0}^{\log_\beta N} \beta^j A_f(D, N \cdot \beta^{-j})\right).    
\end{align}

If $A_f$ is of the form $Cm^\rho$ for $\rho \ge 1$ and for some constant $C$ depending on $D$, then this is equivalent to:
\begin{align*}
    &=  O \left( CN^\rho \sum_{j = 0}^{\log_\beta N} \beta^{j - j\rho} \right)\\
    &= O \left( CN^\rho \sum_{j = 0}^{\log_\beta N} (\beta^{1 - \rho})^j \right)\\
    &= \begin{cases} 
    O(CN^\rho \log_{\beta}N) &\text{if $\rho = 1$}\\
    O \left( CN^\rho \frac{1 - N^{1 - \rho}}{1 - \beta^{1 - \rho}} \right) = O((\frac{1}{1 - \beta^{1 - \rho}} ) CN^\rho) &\text{if $\rho > 1$}
        \end{cases}
\end{align*}

Determining the memory of the index returned by \cref{alg:build_tree} is equivalent to determining the memory of all nodes built on Line~\ref{alg:build_tree:line:get_index} throughout the recursion. Similarly, as long as these calls take longer than constant time, they are the computational bottleneck of the recursion. Thus, we can plug the Vamana build time and memory from~\cite{indyk2023worst} into these results to get the build time and memory of a Vamana WST.

The slow preprocessing version of Vamana takes $O(N^3)$ for construction time and takes up $O(N (4\alpha)^\delta \log \Delta)$ space, so plugging into these results we have that a $\beta$-WST tree with a slow preprocessing Vamana implementation takes $O(\frac{1}{1 - \beta^{-2}} N^3) = O(N^3)$ time to build and has memory size $O((4\alpha)^\delta \log (\Delta) N \log_{\beta} N)$.    
\end{proof}

\runtime*
\begin{proof}

First, we will show that \cref{alg:query} solves the $c$-approximate window search problem. Then, we will show that it solves it in the given running time.

\subsubsection*{\cref{alg:query} solves the $c$-approximate window search problem}

First, we establish correctness and completeness of the evaluated points, i.e., that every point returned has a valid label, and that all points that meet the valid label are "evaluated" on either Line~\ref{alg:query:line:argmin1} or Line~\ref{alg:query:line:ann}.

For correctness, note that if a point is returned on Line~\ref{alg:query:line:argmin1}, by \cref{def:filtered_dataset} it has a label value in $(a, b)$. Similarly, since a point returned on Line~\ref{alg:query:line:ann} is in $D$ and by the if statement we know that all points in $D$ have a label in $(a, b)$, a point returned on Line~\ref{alg:query:line:ann} has a label value in $(a, b)$. Any point returned by Line~\ref{alg:query:return} is an $\argmin$ over points returned in one of these two cases, so we are guaranteed that the overall point $y$ returned has $\ell(y) \in (a, b)$.

For completeness, first consider some call to \cref{alg:query} with $T = (index, (children, sizes), S)$. 
% \julian{$T$ should take 3 arguments, where the second is a pair}
By assumption, $N$ is a power of $\beta$, so we will proceed inductively over $|S|$ equal to powers of $\beta$. Let us first consider any $S$ such that $|S| = 1$. If $x \in S_{(a, b)}$, then it will be evaluated on Line~\ref{alg:query:line:argmin1}. Now we assume that for $|S| = \beta^n$, if $x \in S$, then a call to $Query$ with the tree corresponding to $S$ will evaluate $x$. For all sets $S$ of size $\beta^{n + 1}$ that contain some $x$, if Line~\ref{alg:query:line:argmin1} or Line~\ref{alg:query:line:ann} is executed, then we evaluate $x$. Otherwise, by construction (Line~\ref{alg:build_tree:line:build_tree_recursive} in \cref{alg:build_tree}) the children subsets $S_i$ completely partition $S$, so $x$ is in some $S_i$ with $|S_i| = \beta^{n}$, and so by our inductive hypothesis $x$ will be evaluated when we call query on $children[i]$.

We now show that a $c$-approximate window-filtered nearest neighbor is returned for some (possible recursive) call to $Query$. Because of our correctness guarantee, at some point $q^*$ will be evaluated on Line~\ref{alg:query:line:argmin1} or Line~\ref{alg:query:line:ann}. If $q^*$ is evaluated on Line~\ref{alg:query:line:argmin1}, then because $q^*$ is the closest point to $q$ in all of $D_{(a, b)}$, it will also be the closest point to $q$ in $S_{(a, b)} \subset D_{(a, b)}$, so it will get returned by the $\argmin$ (and $q^*$ is trivially a $c$-approximate window filtered nearest neighbor). If $q^*$ is evaluated on Line~\ref{alg:query:line:ann}, then by the guarantee of the $c$-ANN algorithm $A$, some point $y$ will be returned that is a $c$-ANN to $q$ on $S_{(a, b)}$. Because $q^*$ is also in $S_{(a, b)}$, this implies that $\texttt{dist}_V(q, y) \le c \cdot \texttt{dist}_V(q, q^*)$, so $y$ is also a $c$-approximate window filtered nearest neighbor.

Finally, we show that if any instance of a call to $Query$ finds a $c$-approximate window filtered nearest neighbor, then the overall algorithm will return a $c$-approximate window filtered nearest neighbor. Consider the case that a valid $c$-approximate window filtered nearest neighbor $y$ is returned by Line~\ref{alg:query:line:argmin1} or Line~\ref{alg:query:line:ann}. If this is not a top-level call to $Query$, then $Query$ was called on Line~\ref{alg:query:line:recurse}, so the point $y'$ that gets returned will also be evaluated in the $\argmin$ on Line~\ref{alg:query:return}, and a point $y'$ will be returned from Line~\ref{alg:query:return} that is in $D_{(a, b)}$ (by our correctness result) and has $d(q, y') \le d(q, y) \le c \cdot d(q, q^*)$. Thus by transitivity, $y'$ is also a $c$-approximate window filtered nearest neighbor for $q$, and inductively the point $y''$ that gets returned by the original top-level $Query$ call will be a $c$-approximate window filtered nearest neighbor. 

\subsubsection*{Algorithm $2$ running time}

We will examine each level of the tree built by \cref{alg:build_tree} as \cref{alg:query} traverses it, i.e., the nodes with $|S| = N$, $|S| = N / \beta$, $|S| = N / \beta^2, \ldots, |S| = \beta, |S| = 1$ (the nodes with $|S| = 1$ are just the $index = NULL$ case). 

% At each level $j$, let ${R_j} \subset D$ be the trees that are queried, and let $D_0 = D$ and $D_j = D_{j - 1} - \cup {R_j}$. First, we note that at maximum two nodes in ${R_j}$ will have Query called on them

At a high level, this analysis is similar to $B$-ary segment or Fenwick trees~\cite{pibiri2021practical}, which have $O(\beta \log_\beta N)$ query time and query at most $O(\beta)$ indices per level. The overall idea for our analysis is to show that \cref{alg:query} will run an ANN search on at most $2\beta - 2$ indices per level.

First, we note that our algorithm has a ``one time evaluation guarantee": if we execute an ANN search (Line~\ref{alg:query:line:ann}) or an exact search (Line~\ref{alg:query:line:argmin1}) on some subset $S$, then we did not execute an ANN search or exact search on any parent of $S$ (since then we never would have reached it recursively), so every point in $S$ (and therefore every point in $D_{(a,b)}$) is evaluated just once. 

% \julian{we should explain what the subscripts of the $S$'s mean} \julian{it's not clear how these subsets are formed}
% \julian{I thought we assumed that all $S_i$'s we consider here are a subset of $D_{(a,b)}$} \julian{we should elaborate on why this corresponds to $last - first + 1 > 2\beta - 2$}  \julian{it would be clearer to put the inductive claims in a list and reference each one when you are proving it}  \julian{why fewer than 2?}\josh{@Julian I combined all of your comments up here and have tried to go through and address everything, let me know if anything is still unclear. Part of the problem may be that I have a clear way in my head for how to prove this (basically show that you keep "covering the middle" of the range, and then at each level you cover most of what is left on the left and right), but I don't quite have the infrastructure in the analysis to make it clear that this is what I'm doing.}

Now consider the largest $j$ such that there exists some set $S$ in the tree of size $\beta^j$ such that $S \subset D_{(a, b)}$. In other words, $S$ is the largest set that we built an index for and that entirely consists of points within the filtered dataset corresponding to the query. There may be multiple sets of size $\beta^j$ within $D_{(a, b)}$.
% (we note briefly that $j$ is at least $j = \lceil\log_\beta(N / |D_{(a, b)}|)\rceil + 1$)

Let all sets of size $\beta^j$ be ordered such that each set's labels are strictly less than the next set, and let these sets be indexed by $\{S_i\}$. Let $S_{first}$ be the first set in this ordering that is a subset of $D_{(a,b)}$ and $S_{last}$ be the last set in this ordering that is a subset of $D_{(a, b)}$. 

By completeness, every point in $S_{first}, \ldots, S_{last}$ is evaluated at some level, so the recursive traversal must go through $S_{first}, \ldots, S_{last}$, and since each of these sets is a subset of $D_{(a, b)}$, we will run the ANN search on Line~\ref{alg:query:line:ann} on each of $S_{first}, \ldots, S_{last}$. 

By construction, every $\beta$ sets in $\{S_i\}$ are partitions of a set from level $j + 1$ (e.g., sets $\{S_1, \ldots, S_\beta\}$, $\{S_{\beta + 1}, \ldots, S_{2\beta}\}, \ldots$). Thus, any subsequence of $\{S_i\}$ that is of length  $\ge 2\beta - 1$ must have at least one complete partition of a set from level $j + 1$. Since all of $S_{first}, \ldots, S_{last}$ are subsets of $D_{(a,b)}$, by the one time evaluation guarantee, we know that their parents cannot be subsets of $D_{(a, b)}$ (since then in the recursive traversal, we would have run an ANN search on their parent). Thus, $S_{first}, \ldots, S_{last}$ cannot contain a complete partition of a set from level $j + 1$, so the list $S_{first}, \ldots, S_{last}$ contains fewer than $2\beta - 1$ sets, or equivalently $last - first + 1 \le 2\beta - 2$.

% Consider the sets $S_{first}, \ldots, S_{last}$. By the one time evaluation guarantee, we know that we did not execute an ANN search in any of their parents, so the sets corresponding to their parents were not subsets of $D_{(a, b)}$. Thus, each set in $S_{first}, \ldots, S_{last}$ must have a sibling node that is not a subset of $D_{(a, b)}$, and this sibling node must be within $\beta$ sets in the ordering of $S_i$ (since all of a node's children are in order of label value). \julian{why do we know there are at least $\beta$ $S_i$'s? I think there are fewer than $\beta$ $S_i$'s.}
% Thus, in the case where there are as many sets $S_{first}, \ldots, S_{last}$ as possible \julian{which is $\beta-1$?}, the set $S_{first + \beta - 2}$ \julian{this indexing is confusing.} can have a sibling at $S_{first - 1}$ that is not a subset of $D_{(a, b)}$, and similarly the set $S_{last - \beta + 2}$ can have a sibling at $S_{last + 1}$ that is not a subset of $D_{(a, b)}$. There can be no sets in between these two because such a set would be too far away from a possible sibling that was not a subset of $D_{(a, b)}$, so if we are maximizing $last - first$ we have that $(first + \beta - 2) + 1 = last - \beta + 2$, which implies $last - first + 1 = 2\beta - 2$. Since this represents the maximum $last - first$ can be, we have that $last - first + 1 \le 2\beta - 2$.

% The structure of the tree is that the sets corresponding to a nodes children are a partition of its own set into $\beta$ partitions in order of label values.

We also know that at level $j$, there are at most two more sets, $S_{first - 1}$ and $S_{last + 1}$, that have a non-empty intersection with $D_{(a, b)}$ (these are the sets that potentially contain points with labels just larger and just smaller than $a$ and just larger and just smaller than $b$). 

This leads to our inductive hypothesis, which has three claims: 
\begin{enumerate}
    \item For all levels with $j' \le j$ (i.e., with $|S| = \beta^{j'}$), there can be at most two sets that have a non-empty intersection with $D_{(a, b)}$ that are not fully evaluated (i.e., that we recurse into).
    \item No set that we recurse into or evaluate on level $j' \le j$ is a superset of $D_{(a,b)}$. 
    % \julian{is the sign reversed?}
    \item We will run the ANN search on Line~\ref{alg:query:line:ann} at most $2\beta - 2$ times on each level.
\end{enumerate}

We have just shown the base case for $j' = j$. Now consider some $0 \le j' < j$. By part $1$ of the inductive hypothesis, we recurse into $2$ or fewer sets  on level $j' + 1$ that have a nonempty intersection with $D_{(a, b)}$. For each of these sets, at most $\beta - 1$ of its children will be a subset of $D_{(a, b)}$ (if all $\beta$ of its children were a subset of $D_{(a, b)}$, then the set itself would have been a subset of $D_{(a, b)}$ and would have been fully evaluated), and thus at most $2(\beta - 1) = 2\beta - 2$ ANN searches are run on level $j'$. This proves part $3$ of our inductive claim. Furthermore, since each of the at most two sets that we are recursing into are not a superset of $D_{(a, b)}$ by inductive claim $2$, there can only be at most one side of each of their label ranges that expand beyond $(a, b)$. Thus, when we partition each of these sets, only one child of each of the sets can have a label range that overlaps $(a, b)$; the rest will either have labels entirely in $(a, b)$ or entirely outside of $(a, b)$. Thus there will be at most two 
% \julian{the previous sentence implies there is at most 1?} \josh{I clarified it's one per parent $S$} 
sets that have a nonempty intersection with $D_{(a, b)}$ that are not fully evaluated, proving inductive claim $1$. Finally, because each set that we recurse into on level $j'$ is not a superset of $D_{(a, b)}$, all of the children we recurse into that are subsets of these sets are also not supersets of $D_{(a, b)}$, proving inductive claim $2$. 

We do work in \cref{alg:query} on Line~\ref{alg:query:line:ann}, the loop on Line~\ref{alg:query:line:loop}, and Line~\ref{alg:query:return} (we do not do work on Line~\ref{alg:query:line:argmin1} because the $\argmin$ is just over one point; the $\argmin$ is necessary for the more general case where $N$ is a not a power of $j$ and the leaf nodes may have more than one point). 
% \julian{if the argmin is over one point, why do we do an argmin?} 
As a direct result from the first part of our inductive claim, we have that we will only make it to the loop on Line~\ref{alg:query:line:loop} and the $\argmin$ on Line~\ref{alg:query:return} twice for each of the $\log_\beta(N)$ levels. The time complexity for the loop is $O(\beta)$, and the time complexity for Line~\ref{alg:query:return} is $O(\beta d)$ because the maximum size for the candidate list is $\beta$, and for each candidate in the list we spend $O(d)$ doing a distance computation.  Also, from the third part of our inductive claim, we have that we call ANN search on an index of size $m = \beta^j$ at most $2\beta - 2$ times for all $j \in {1, \ldots, \log_\beta(N)}$. Finally, again from the third part our inductive claim, we evaluate at most $2\beta  - 2$ sets of size $1$, so we evaluate Line~\ref{alg:query:line:argmin1} a maximum of $2\beta  - 2$ times. This gives the following total runtime for \cref{alg:query}:
$$O\left(\log_\beta(N) * (\beta + \beta d) + (2\beta - 2) * \sum_{j = 0}^{\log_\beta(N)} A_q(D, N * \beta^{-j})\right) = O\left(\beta\log_\beta(N)d + \beta * \sum_{j = 0}^{\log_\beta(N)} A_q(D, N * \beta^{-j})\right).$$
 \end{proof}

\commonfuncclasses*
 \begin{proof}
 % First, we note that if we have a term $f(N) + b* g(N)$, and $f(N)$ is $O(g(N))$, then $a + b * f(N)$ is $O(f(N) + b * g(N))$; \julian{are $a$ and $b$ constants here? we should clarify what they are. also we should use different variable names, since we used $a$ and $b$ for something else (I'm also not sure where we used this fact below)} \josh{Maybe we can just leave this out entirely then? I used it on e.g. the first line to plug in the runtime for A and drop the big O} see~\cite{StackOverflowNestedBigO} for a proof of this fact. 
 For the first result, we have via substitution into \cref{thm:runtime}:
\begin{align*}
    &O\left(\beta\log_\beta(N)d + \beta\sum_{j=0}^{\log_\beta{N}} A_q(D, N \cdot \beta^{-j})\right)\\
    &= O\left(\beta\log_\beta(N)d + \beta\sum_{j=0}^{\log_\beta{N}} CdN^\rho \cdot \beta^{-j \rho}\right)\\
        &= O\left(\beta\log_\beta(N)d + C\beta dN^\rho\sum_{j=0}^{\log_\beta{N}} \beta^{-j \rho})\right)\\
        &= O\left(\beta\log_\beta(N)d + C\beta dN^\rho \frac{1}{1 - \beta^{-\rho}}\right)\\
        &= O\left(\frac{C\beta dN^\rho}{1 - \beta^{-\rho}}\right).
\end{align*}
and for the second result, we have via substitution into \cref{thm:runtime}: 
\begin{align*}
    &O\left(\beta\log_\beta(N)d + \beta\sum_{j=0}^{\log_\beta{N}} A_q(D, N \cdot \beta^{-j})\right)\\
    &= O\left(\beta\log_\beta(N)d + \beta\sum_{j=0}^{\log_\beta{N}} A_q(D)\right)\\
        &= O\left(\beta\log_\beta(N)d + \beta\log_\beta(N)A_q(D)\right)\\
        &= O\left(\beta\log_\beta(N)(d + A_q(D))\right).
\end{align*}
\end{proof}

\existingannruntime*
\begin{proof}
\cite{indyk2023worst} consider greedy search on an $\alpha$-Vamana graph built on a dataset $D$ with "slow preprocessing." They show that the search procedure is guaranteed to return a $(\frac{\alpha + 1}{\alpha - 1} + \epsilon)$-ANN in $O(\log_\alpha(\frac{\Delta}{(\alpha - 1) \epsilon}))$ steps, each of which take $O((4\alpha)^\delta \log \Delta)$ time. See \cref{app:diskann} for an explanation of $\alpha$ and a complete overview of Vamana. 

We can multiply together these two bounds on the running times to get an upper bound on the running time of the entire procedure:
$$O\left(\log_\alpha\left(\frac{\Delta}{\left(\alpha - 1\right) \epsilon}\right)\left(4\alpha\right)^\delta \log \Delta\right) = O\left(\log_\alpha\left(\frac{\Delta}{\left(\alpha - 1\right) (c - \frac{\alpha + 1}{\alpha - 1})}\right)\left(4\alpha\right)^\delta \log \Delta\right).$$
Note the equality is due to the fact that $c = \frac{\alpha + 1}{\alpha - 1} + \epsilon$. 

Now consider some $S \subset D$. We claim that the doubling dimension $\delta$ and the aspect ratio $\Delta$ for $S$ are no greater than for $D$.~\cite{indyk2023worst} describe doubling dimension as the the minimum value $\delta$ such that for any ball of radius $r$ centered at some point $x$ in $D$, $2^\delta$ balls of radius $r / 2$ can be arranged to cover all points in $D \cap B(x, r)$ (many previous works use the same or an extremely similar definition of doubling dimension (see, e.g.,~\cite{clarkson1997nearest, krauthgamer2004navigating}). Because $S$ is a subset of $D$, any covering of $D \cap B(x, r)$ is also a covering of $S \cap B(x, r)$ for all $x$ in $D$ (and also trivially therefore all $x$ in $S \subset D$), so we know that the doubling dimension of $S$ is at most the doubling dimension of $D$. Similarly,~\cite{indyk2023worst} use the aspect ratio of $D$, which is
$$\frac{\max_{x_1, x_2 \in D, x_1 \ne x_2} \texttt{dist}_V(x_1, x_2)}{\min_{x_1, x_2 \in D, x_1 \ne x_2} \texttt{dist}_V(x_1, x_2)}.$$
For any subset $S$ of $D$, let the two points corresponding to the smallest distance be $x_{small}$ and $y_{small}$ and the two points corresponding to the largest distance be $x_{large}$ and $y_{large}$. Since $S \subset D$, $x_{small}, y_{small}$ are also in $D$, so the smallest distance in $D$ is less than or equal to $\texttt{dist}_V(x_{small}, y_{small})$, and similarly $x_{large}, y_{large}$ are also in $D$, so the largest distance in $D$ is greater than or equal to $\texttt{dist}_V(x_{large}, y_{large})$. Thus compared to $\Delta$, the numerator of $\Delta_S$ is no larger and the denominator is no smaller, and so $\Delta_S \le \Delta$. In other words, the aspect ratio of any subset $S$ is less than the aspect ratio $\Delta$ of $D$. 

Because our running time result above is monotonic in $\Delta$ and $\delta$, and the other parameters $\alpha$ and $c$ are constant, we have that $c$-ANN search on any subset $S \subset D$ is upper bounded by the running time on the entire dataset $D$. Thus, since $O(A_q(D, m) = A_q(D))$, we can now plug in to \cref{lem:commonfuncclasses}, giving us our final result.
\end{proof}

\wstblowup*
\begin{proof}
To see that the worst case blowup factor is $O(N)$, consider the first split of $D$ into children $S_i$ for $i = 1, \ldots \beta$. These $S_i$ correspond to label ranges $\{[a_i, b_i]\}$, where $a_0 = 0$, $b_{last} = N$, and each $a_i = b_{i - 1} + 1$. Consider the range $(b_1, a_2)$. This range is not a subset of any $S_i$. Furthermore, because all smaller ranges further down the tree are strict subsets of some $S_i$, this range is also not a subset of any smaller range. Thus the smallest range that $S_i$ is a subset of is the top level range, so a $\beta$-WST has a worst case blowup of $\frac{N}{2} = O(N)$.

For the cost, we note that the label ranges of each level of the tree (except possibly the last, since it might be only partially full) are a partition of $\{1,\ldots,N\}$, so the sum of $b_i - a_i$ is equal to $N$ for all levels but the last. There are $\lfloor\log_\beta(N)\rfloor$ levels, and one possibly non-full level at the bottom of the tree which is smaller than or equal to a full partition of $\{1,\ldots,N\}$ and so has a cost less than or equal to $N$. Thus the total cost is bounded by $\lceil\log_\beta(N)\rceil \cdot N$.
% \julian{I think we should just say the cost is at most $\lceil\log_\beta(N)\rceil \cdot N$, which will get rid of the inconsistency in the lemma statement and slightly simplify the discussion}
\end{proof}

\begin{figure*}[t]
\begin{center}
\centerline{\includegraphics[width=9cm]{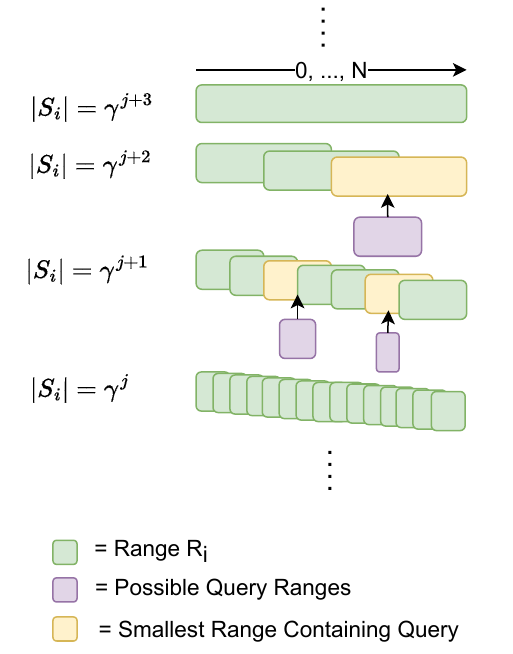}}
\caption{Illustration of structure of ranges for~\cref{thm:betterblowup}. }
\label{fig:theory-blowup-help}
\end{center}
\end{figure*}

\betterblowup*
\begin{proof}
At a high level, our approach is to devise a strategy that can ensure all sets of size $m$ have a blowup factor of $2$. We will then repeat this strategy for $m = \gamma^j$ for all possible powers of $j$, which will ensure that all possible ranges have a small worst case blowup factor. For a diagram of this structure see~\cref{fig:theory-blowup-help}.

First, consider the problem of choosing ranges $R_i$ such that every range of size $m$ is a subset of some $R_i$ with blowup factor equal to $2$. One approach is to choose ranges of $$cover(m) = \{[jm + 1, (j + 2)m ] \ |\ j \in \mathbb{Z}_{\ge 0}, (j + 2)m \le N\} \cup {[N - 2m + 1, N]}.$$

The ends of the ranges start at $2m$ and go until $N$ by multiples of $m$, for a total of $\lfloor \frac{N}{m} \rfloor - 1$ ranges. These, plus the additional range $[N - 2m + 1, N]$, lead to a total of $\lfloor \frac{N}{m} \rfloor$ ranges created using this strategy.  Each range has width $2m$, so the arrangement has cost $$\left\lfloor \frac{N}{m } \right\rfloor \cdot 2m \le 2N.$$
% \julian{what about the last range?}

% \julian{it's not clear whether we are creating these ranges for one specific $m$, for all possible $m$'s, or some other way. can we explicitly state what the algorithm is? it may also be helpful to include a figure}
We now show that these ranges $R_i$ do indeed cover all ranges of size $m$ with blowup factor equal to $2$. Consider some range of length $m$ starting at $a$. If $a$ is within the first $m + 1$ integers in a range, then it is entirely within the range. Therefore, we are interested in the union of the first $m + 1$ integers in all of the ranges, or
$$\bigcup_{(j + 2)m \le N} \left\{[jm + 1, (j + 1) m + 1]\right\} \bigcup [N - 2m + 1, N - m + 1] \supset [1, N - 2m] \bigcup [N - 2m + 1, N - m + 1] = [1, N - m + 1].$$
This is all possible starting points for a range of length $m$, so $R_i$ does indeed cover all ranges of size $m$. Furthermore, each range is of size $2m$, so the blowup factor for these ranges of size $m$ is $2$. 

% Note that as a followup, this also shows that any range of size $m / C$ is a subset of some $R_i$ with a blowup factor equal to $2C$. \julian{since we removed $C$ from the theorem, do we still need this statement?}

% \julian{it seems that we are doing the process above for all values of $m$ that are powers of $(C\gamma)^j$?} 

Now consider $R = \{cover(\gamma^j) \ |\ j \in \mathbb{Z}_{\ge 0}, \gamma^j < N\} \cup {(0, N)}$. We have that 
% \julian{unclear what the two cases on the first line below means} \josh{The idea is that the additional set in the union only matters if $N$ is not a power of $\gamma$, but I can remove it since it may be needlessly confusing}
\begin{align*}
    cost &= \lfloor \log_{\gamma}(N) \rfloor \cdot 2N + N\\
    &\le 2N\log_{\gamma}(N) + N\\
    &= N\left(2\log_\gamma(N) + 1\right).
\end{align*}
% \julian{why does the $C$ go from the base of the log to the denominator?}

Furthermore, we now show that $R$ has worst case blowup factor $2\gamma$.
% on $[N]$ \julian{not sure what "on $[N]$" means}. 
Consider some range $r_q$ of size $m$. Consider the minimum $j$ such that $\gamma^j$ is greater than $m$. Let us first consider the case when $\gamma^j$ is less than $N$. Consider some range $r$ of size $\gamma^j$ that contains $r_q$. By the definition of $cover(\gamma^j)$, there is some range of size $2\gamma^j$ that contains $r$ and that therefore contains $r_q$. Furthermore, since this $j$ is the minimum $j$ such that $\gamma^j > m$, we have that $m > \gamma^{j - 1}$. Thus the maximum blowup factor for $r_q$ is less than $2\gamma^j / \gamma^{j - 1} = 2\gamma$. In the case where $\gamma^j \ge N$, the smallest containing range is $(0, N)$. We have that $m > \gamma^{j - 1} \ge N / \gamma$, so $N / m < \gamma$ and thus the blowup factor for $r_q$ is less than $\gamma$ (and also less than $2\gamma$).

\end{proof}

\section{ChatGPT Queries for RedCaps Query Generation}
\label{app:redcaps_gpt}
Queries were generated in two sessions with ChatGPT-4. The first consisted of the first query and the second query repeated 4 times. The second consisted of 100 examples copied from the first session and the third and four query. The first query:
\begin{quote}
I wish to run an experiment where I generate many possible text queries for my image search system. Can you help me generate queries? I want you to make the queries casual, and be as varied and creative as possible. To the best of your ability, don't repeat yourself! Here are a few example queries:
"Funny cat memes"
"C++ coding joke"
"Vegan recipe with blueberries".    
\end{quote}
Repeated second query:
\begin{quote}
    Can you generate 100 more? And still make sure to be as creative and casual as possible, and don't repeat things you've already said! Additionally, try to use unique semantic and syntactic structure where possible.
\end{quote}
Third and fourth queries:
\begin{quote}
    Try not to generate queries with locations in them, because I already have lot's of those. Thank you!
\end{quote}
and 
\begin{quote}
    Great, now I just need 100 more, again as little repeats as possible, be creative! These can be more "internet" language, so things like, e.g., "funny cat memes."
\end{quote}

\section{Vamana Primer}
\label{app:diskann}

Vamana~\cite{diskann} is an approximate nearest neighbor algorithm that builds a graph on the input dataset $D$. The entire system from \cite{diskann}, including checkpointing and error recovery, is called DiskANN; Vamana is solely the in-memory ANNS component.

To construct the ``slow-preproccessing" variant, which is the variant with theoretical guarantees from~\cite{indyk2023worst}, for each point $x$, we first connect $x$ to all points. We then sort all other points in the graph in terms of increasing distance from $x$. Starting from the first point $y$ in the list, we prune edges from $x$ to all other points $y'$ where $\alpha \cdot \texttt{dist}_V(y, y') \le \texttt{dist}_V(x, y)$. We repeat this pruning process with the next closest unpruned point until we reach the end of the list; the unpruned points are the neighbors of $x$.

To construct the ``fast-preprocessing" variant, which is the variant used in practice, we start with an empty graph. We then do a beam search query for the nearest neighbor for all points $x \in D$, twice, building up the graph as we go. For each  search, we record all points traversed in the beam search, along with the nearest neighbor if it was not found, and then add these as neighbors to $x$. We then prune this list using the same  heuristic we use for the ``slow-preprocessing" variant, and also enforce with a hard cutoff that there are at maximum $\degree{}$ number of neighbors.

A beam search of size $B$ is a generalization of a greedy search. Given a query $x$, we start at a start node $s$ and ``explore" $s$ by adding neighbors of $s$ to a queue. This queue consists of the closest $B$ points to $x$ we have seen so far, explored or unexplored. We continually explore the closest unexplored node from the queue to $x$ until all nodes in the queue are explored. By increasing $B$, the beam search explores more points and is more likely to find a better nearest neighbor of $x$. We do beam searches for ``fast preprorocessing" index construction and for approximate nearest neighbor queries.

\section{Experiments}\label{sec:experiments-appendix}

\subsection{Dataset License Information}
\label{app:licensing}
\datasetname{SIFT}, \datasetname{GloVE}, and \datasetname{DEEP} are released under an MIT license by the ANN benchmarks repository~\cite{aumuller2020ann}. The original RedCaps license has a
restriction to non-commercial use, so our modified \datasetname{Redcaps} dataset
is released under the same restriction. Since we generated the \datasetname{Adverse} dataset ourselves, we release it under an MIT license.

\subsection{Pareto Frontiers}\label{sec:pareto}
See \cref{fig:glove_results}, \cref{fig:sift_results}, and \cref{fig:redcaps_results}  for full Pareto frontier plots for a representative sample of filter widths on all datasets not included in the main text.
% \julian{this is just for glove and sift. we don't have redcaps?} \josh{Just running VBASE on redcaps overnight, I'll add in the morning}.

\begin{figure*}[t]
\begin{center}
\centerline{\includegraphics[width=16cm]{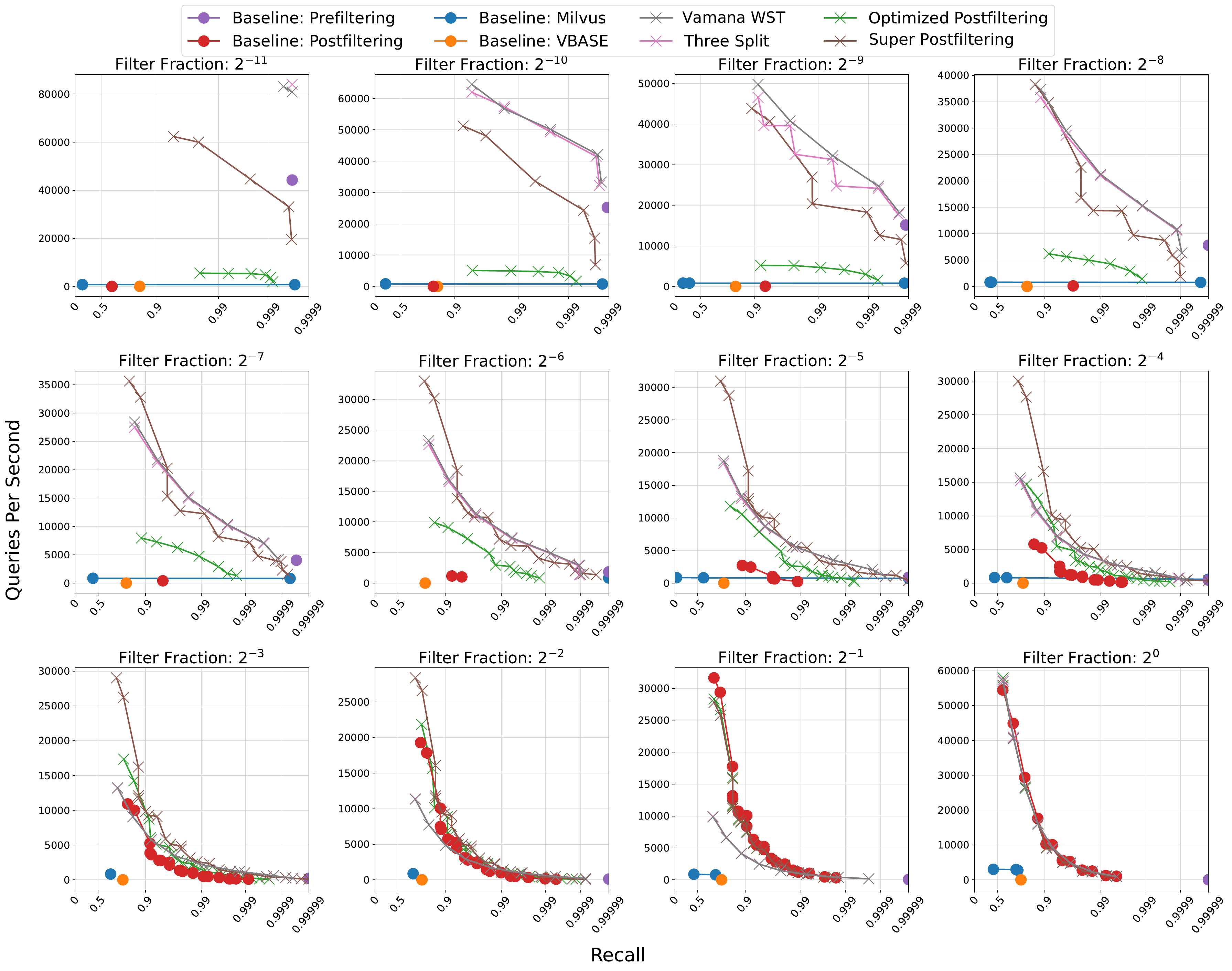}}
\caption{Comparison of Pareto frontiers of all methods on window search with different filter fractions on \datasetname{GloVe}. Up and to the right is better. On the medium filter fraction settings, our methods achieve multiple orders of magnitude more
queries per second than the baselines at the same recall levels. All methods are run with $16$ threads.}
\label{fig:glove_results}
\end{center}
\end{figure*}

\begin{figure*}[t]
\begin{center}
\centerline{\includegraphics[width=16cm]{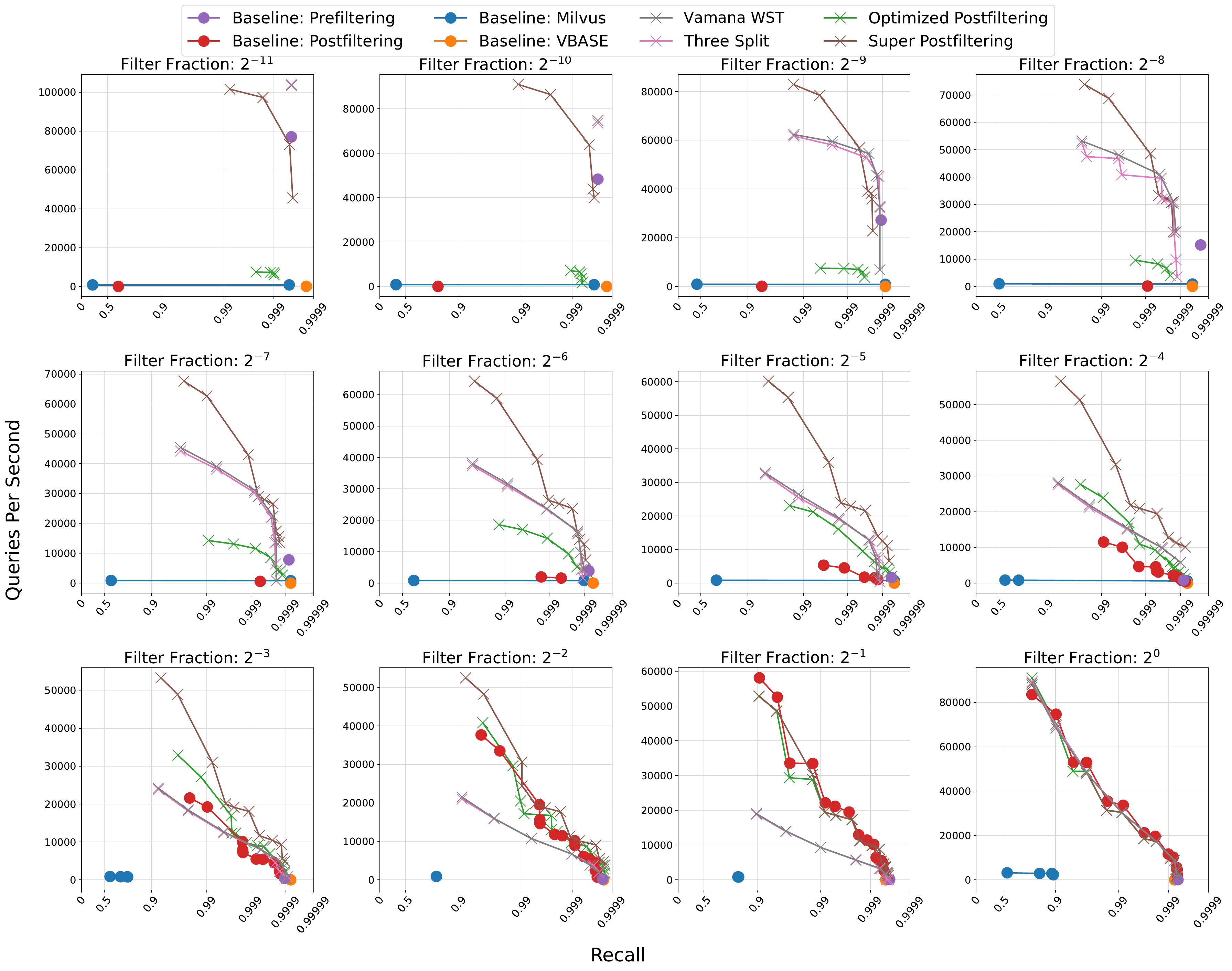}}
\caption{Comparison of Pareto frontiers of all methods on window search with different filter fractions on \datasetname{SIFT}. Up and to the right is better. On the medium filter fraction settings, our methods achieve multiple orders of magnitude more
queries per second than the baselines at the same recall levels. All methods are run with $16$ threads. }
\label{fig:sift_results}
\end{center}
\end{figure*}

\begin{figure*}[t]
\begin{center}
\centerline{\includegraphics[width=16cm]{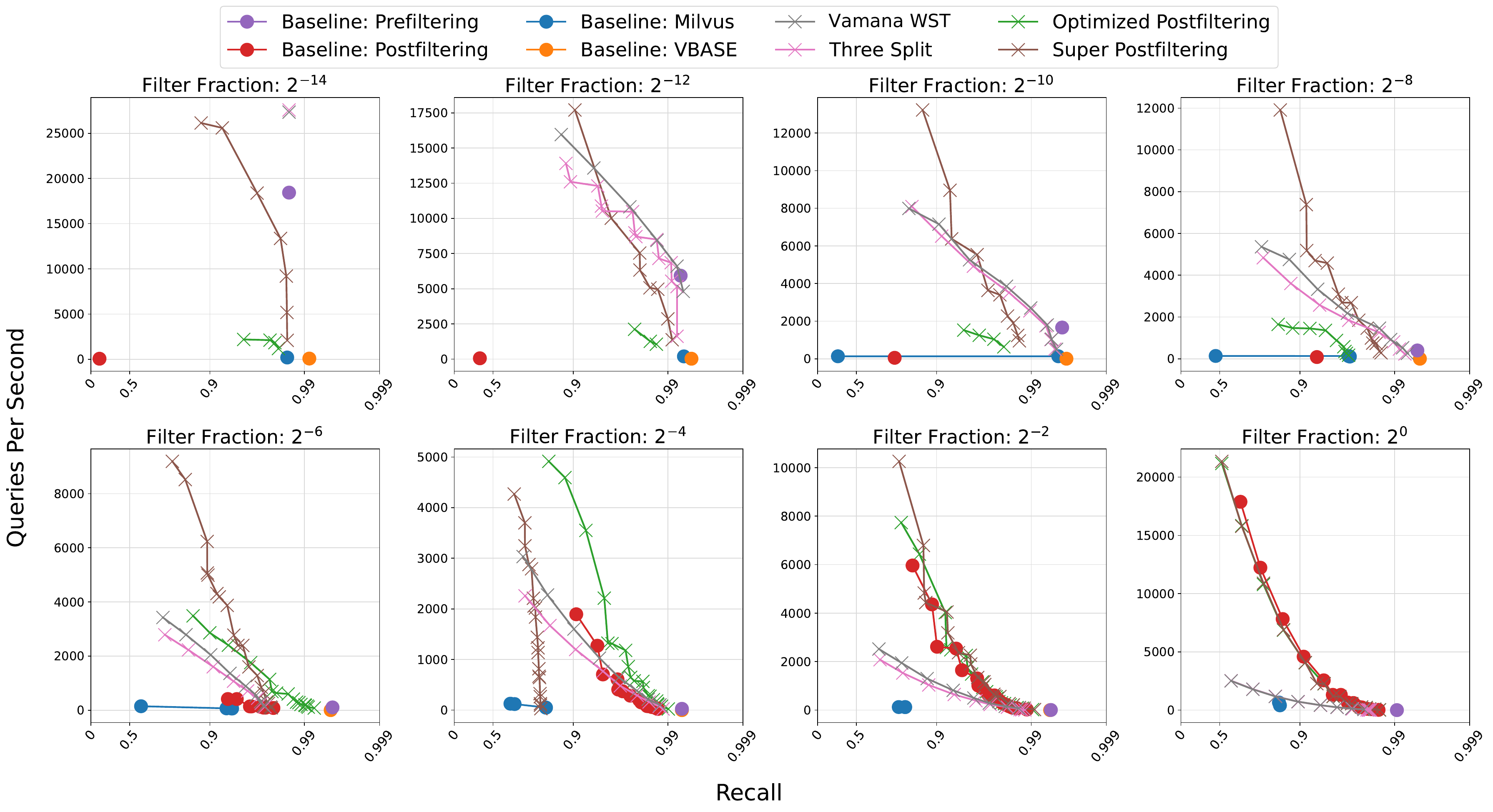}}
\caption{Comparison of Pareto frontiers of all methods on window search with different filter fractions on \datasetname{Redcaps}. Up and to the right is better. On the medium filter fraction settings, our methods achieve multiple orders of magnitude more
queries per second than the baselines at the same recall levels. All methods are run with $16$ threads. }
\label{fig:redcaps_results}
\end{center}
\end{figure*}

\subsection{Comments on Memory and Performance}
We expect our method to scale well to larger datasets because our theory predicts only a $\log_\beta N$ factor increase in memory cost over a single Vamana index constructed on $D$. Our implementation also addresses memory size further by not constructing tree nodes that represent subsets smaller than $1000$ points and only storing the dataset $D$ once. \cref{fig:branching_buildtime_memory} shows the memory and construction time of a $\beta$-WST tree constructed for SIFT as we increase $\beta$. For large $\beta$, a $\beta$-WST tree has an only slightly larger memory footprint (about 2X) than a single ANN
index, and for the theoretical and implementation reasons above we expect this trend to hold for larger
datasets.

Finally, we run additional performance experiments. See~\cref{fig:full_branching_filter_widths} for comparisons of varying $\beta$ on \datasetname{SIFT} across all filter fractions, and see~\cref{tab:all_speedups} for speedups of our best method over the best baseline across recall levels $0.8$, $0.9$, $0.99$, and $0.995$.

% \begin{table}[h]
%     \centering
%     \begin{tabular}{|c|c|}
%         \hline
%         Method & Memory \\
%         \hline
%         $2$-WST tree & 4.8 GB \\
%         $4$-WST tree & 3.3 GB \\
%         $8$-WST tree & 2.9 GB \\
%         $16$-WST tree & 2.6 GB \\
%         $32$-WST tree & 2.4 GB \\
%         $64$-WST tree & 2.3 GB \\
%         Single DiskANN index & 1.0 GB \\
%         \hline
%     \end{tabular}
%     \caption{Memory usage of $\beta$-WST tree for SIFT as $\beta$ increases.
%     \julian{can we reference this figure in body? or if it's redundant with Figure 9 we can remove it)}
%     }
%     \label{tab:memory}
% \end{table}

\begin{figure}
    \centering
    \includegraphics[width=10cm]{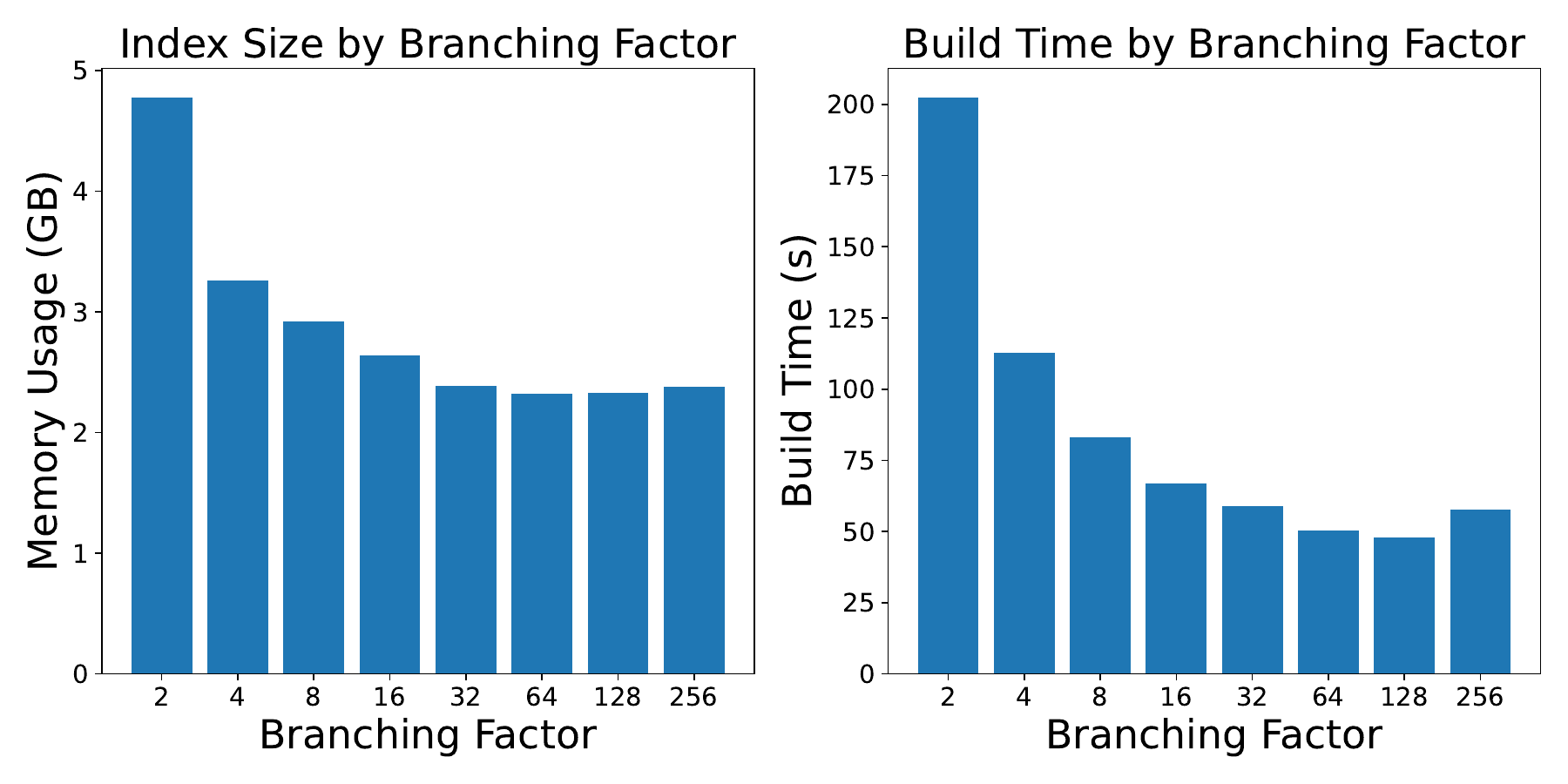}
    \caption{Plots of index size and build time for varying branching factors $\beta$ for \algname{Vamana WST} on \datasetname{SIFT}. The indices were built using 96 threads. Smaller values are better. }
    \label{fig:branching_buildtime_memory}
\end{figure}

\begin{figure}
    \centering
    \includegraphics[width=16cm]{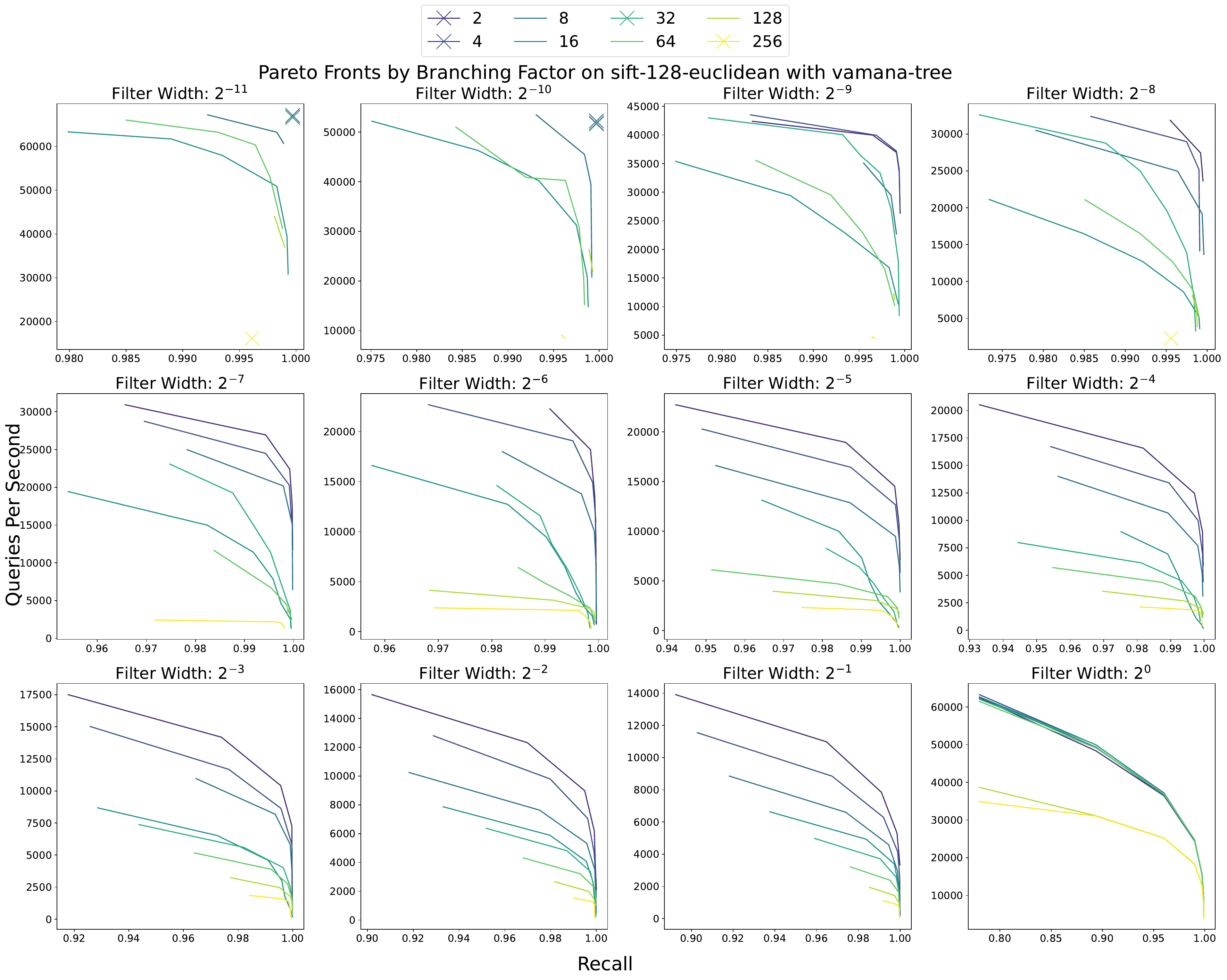}
    \caption{Pareto curves of recall vs.\ throughput on \datasetname{SIFT} for varying window sizes and branching factors $\beta$ for \algname{Vamana WST}. The experiment was run using 16 threads. Up and to the right is better.}
    \label{fig:full_branching_filter_widths}
\end{figure}

\begin{table*}
\centering
\caption{Speedup of our best method over the best baseline, restricted to hyper-parameter settings that yield the recall in parenthesis in the dataset column. N/A means that none of our methods achieved that recall (on \datasetname{Redcaps} this is due to poor Vamana graph quality).}
\begin{tabular}{|l|rrrrrrrrrrrr|}
\toprule
Dataset & $2^{-11}$ & $2^{-10}$ & $2^{-9}$ & $2^{-8}$ & $2^{-7}$ & $2^{-6}$ & $2^{-5}$ & $2^{-4}$ & $2^{-3}$ & $2^{-2}$ & $2^{-1}$ & $2^{0}$ \\
\hline
\datasetname{Deep} (0.8) & 10.49 & 18.46 & 35.65 & 65.40 & 84.88 & 26.23 & 10.23 & 4.59 & 2.37 & 1.33 & 0.78 & 0.79 \\
\datasetname{SIFT} (0.8) & 1.35 & 1.88 & 3.05 & 4.87 & 8.68 & 16.51 & 11.26 & 4.92 & 2.47 & 1.39 & 0.91 & 0.94 \\
\datasetname{GloVe} (0.8) & 1.90 & 2.56 & 3.29 & 4.90 & 8.82 & 16.37 & 10.57 & 4.77 & 1.49 & 0.90 & 0.90 & 0.92 \\
\datasetname{Redcaps} (0.8) & 5.10 & 7.95 & 11.86 & 29.94 & 36.46 & 20.69 & 5.89 & 2.59 & 3.27 & 1.14 & 0.88 & 0.88 \\
\hline
\datasetname{Deep} (0.9) & 10.49 & 18.46 & 35.65 & 65.40 & 84.88 & 26.23 & 10.23 & 4.26 & 2.18 & 1.23 & 0.75 & 0.76 \\
\datasetname{SIFT} (0.9) & 1.35 & 1.88 & 3.05 & 4.87 & 8.68 & 16.51 & 11.26 & 4.92 & 2.47 & 1.39 & 0.91 & 0.94 \\
\datasetname{GloVe} (0.9) & 1.90 & 2.56 & 3.29 & 4.46 & 5.38 & 9.97 & 6.96 & 4.04 & 1.87 & 1.57 & 0.90 & 0.90 \\
\datasetname{Redcaps} (0.9) & 3.18 & 5.38 & 8.92 & 18.55 & 19.94 & 10.46 & 4.33 & 1.87 & 1.70 & 1.55 & 0.89 & 0.89 \\
\hline
\datasetname{Deep} (0.99) & 6.80 & 11.77 & 22.05 & 40.06 & 50.55 & 9.86 & 4.33 & 3.70 & 1.58 & 1.47 & 0.75 & 0.75 \\
\datasetname{SIFT} (0.99) & 1.35 & 1.79 & 2.88 & 4.53 & 8.04 & 10.10 & 6.73 & 2.88 & 1.61 & 1.01 & 0.88 & 0.91 \\
\datasetname{GloVe} (0.99) & 1.90 & 1.99 & 2.13 & 1.96 & 3.03 & 4.02 & 6.00 & 5.71 & 4.66 & 2.67 & 0.95 & 0.92 \\
\datasetname{Redcaps} (0.99) & 1.30 & 1.08 & 1.50 & 1.32 & 1.36 & 1.87 & 3.11 & 0.86 & 1.82 & 3.75 & N/A & N/A \\
\hline
\datasetname{Deep} (0.995) & 6.80 & 11.77 & 22.05 & 26.07 & 32.98 & 11.31 & 9.91 & 2.41 & 1.98 & 1.49 & 0.75 & 0.78 \\
\datasetname{SIFT} (0.995)  & 1.35 & 1.79 & 2.88 & 3.20 & 5.51 & 10.10 & 6.73 & 2.16 & 1.98 & 0.96 & 0.89 & 0.88 \\
\datasetname{GloVe} (0.995)  & 1.90 & 1.99 & 1.63 & 1.96 & 2.56 & 3.28 & 3.93 & 4.64 & 4.58 & 2.97 & 0.94 & 0.93 \\
\datasetname{Redcaps} (0.995)  & N/A & 0.30 & N/A & N/A & N/A & N/A & N/A & N/A & N/A & N/A & N/A & N/A\\
\hline
\end{tabular}
\label{tab:all_speedups}
\end{table*}

\end{document}